\documentclass[aps,
reprint,prx,
superscriptaddress,
twocolumn,
floatfix]{revtex4-2}
\usepackage[toc,page]{appendix}
\usepackage{amsmath,mathtools,amsthm,amssymb}
\usepackage{qcircuit}
\usepackage{units}
\usepackage{complexity}
\usepackage{graphicx}
\usepackage{color}
\usepackage{xcolor}
\usepackage{bbold}

\usepackage[linesnumbered,ruled,vlined]{algorithm2e}
\usepackage{listings}

\usepackage{enumitem}
\usepackage{pifont}
\usepackage{times}
\usepackage{comment}
\usepackage{array}
\newcolumntype{H}{>{\setbox0=\hbox\bgroup}c<{\egroup}@{}}
\usepackage{graphics}
\usepackage[normalem]{ulem}

\usepackage{pifont}
\usepackage{ragged2e}
\usepackage{complexity}
\newcommand{\mc}[1]{\mathcal{#1}}
\newcommand{\norma}[1]{ \lVert  {#1} \rVert }
\newcommand{\absa}[1]{ \lvert  {#1} \rvert }
\definecolor{myrefcolor}{RGB}{242, 10, 10}
\definecolor{myurlcolor}{RGB}{255, 138, 48}

\newcommand{\wtp}{\bar W_t}
\newcommand{\wtpd}{\bar W_t^\dagger}
\usepackage[
    breaklinks,
    pdftex,
    colorlinks=true,
    linkcolor=myrefcolor,
    citecolor=myurlcolor,
    urlcolor=myurlcolor
]{hyperref}

\usepackage[toc,page]{appendix}

\usepackage{braket}
\usepackage{physics}
\usepackage{csquotes}
\usepackage{thmtools, thm-restate}

\newtheorem{theorem}{Theorem}

\newtheorem{lemma}[theorem]{Lemma}
\newtheorem{corollary}[theorem]{Corollary}

\usepackage{float}

\makeatletter
\renewcommand{\fnum@figure}{FIG. \thefigure}  
\renewcommand{\@floatboxreset}{\raggedright}
\makeatother

\begin{document}

\title{Efficiently learning fermionic unitaries with few
non-Gaussian gates}

\date{\today}

\author{Sharoon Austin}
  \affiliation{Joint Center for Quantum Information and Computer Science, NIST/University of Maryland, College Park, MD, 20742, USA} 
  \affiliation{Joint Quantum Institute, NIST/University of Maryland, College Park, MD, 20742, USA}

\author{Mauro E.S. Morales}\affiliation{Joint Center for Quantum Information and Computer Science, NIST/University of Maryland, College Park, MD, 20742, USA}  

\author{Alexey Gorshkov}
\affiliation{Joint Center for Quantum Information and Computer Science, NIST/University of Maryland, College Park, MD, 20742, USA}
\affiliation{Joint Quantum Institute, NIST/University of Maryland, College Park, MD, 20742, USA}

\begin{abstract}
Fermionic Gaussian unitaries are known to be efficiently learnable and simulatable. In this paper, we present a learning algorithm that learns an $n$-mode circuit containing $t$ parity-preserving non-Gaussian gates. While circuits with $t = \textrm{poly}(n)$ are unlikely to be efficiently learnable, for constant $t$, 
we present a polynomial-time algorithm for learning the description of the unknown fermionic circuit within a small diamond-distance error.  
Building on work that studies the state-learning version of this problem, our approach relies on learning approximate Gaussian unitaries that transform the circuit into one that acts non-trivially only on a constant number of Majorana operators. Our result also holds for the case where we have a qubit implementation of the fermionic unitary.
\end{abstract}
\maketitle

\section{INTRODUCTION}

The task of learning unknown quantum unitaries is fundamental to quantum information science \cite{Chuang_1997, PhysRevA.77.032322, PhysRevLett.86.4195}. This task is important for the development of quantum algorithms and the characterization of quantum devices. However, learning unitaries of arbitrary gate complexity is exponentially hard \cite{OBrien04, kothopt}, making it crucial to identify classes of unitaries that can be learned efficiently. Moreover, practical quantum computation requires verification and validation of quantum circuits with bounded gate complexity. From an experimental standpoint, benchmarking and calibrating quantum devices can be thought of as a learning problem. Previous works have focused on providing learning algorithms in various 
scenarios, including circuits with Clifford unitaries together with a constant number of $T$ gates \cite{CliffordsT}, unitaries with a constant
number of two-qubit gates \cite{Zhao_2024}, and quantum circuits of constant depth \cite{yunchao}.

This work focuses on learning fermionic unitaries \cite{Jozsa_2008,terhalferm}. We consider two cases.  The first case, which we refer to as the \textit{fermionic implementation},  is where the unknown unitary is implemented on fermionic modes.
To define our learning problem, it is important to distinguish Gaussian and non-Gaussian unitaries. In the fermionic implementation, 
we will restrict 
Gaussian unitaries to be parity-preserving, which corresponds to time evolution under quadratic fermionic Hamiltonians, allowing for both hopping and pairing terms  
\cite{bravyi,terhalferm}.
The learning algorithm in this case must use input states that can be efficiently prepared on a fermionic quantum computer using parity-preserving gates.
The second case, which we refer to as the \textit{qubit implementation}, is where the unknown unitary is implemented on a chain of qubits (which can be related to a chain of fermionic modes via the Jordan-Wigner transformation). In this qubit implementation, the allowed Gaussian unitaries  are not required to preserve parity. They  
are generated by nearest-neighbor matchgates (which are parity-preserving and map to fermionic Gaussian gates under the Jordan-Wigner transformation) and $X_1$, the Pauli $X$ matrix on the first qubit \cite{Valiant, Wan_2023}.

It is well-established that Gaussian unitaries in both implementations can be learned efficiently \cite{mauro,cudby2024learninggaussianoperationsmatchgate}. Throughout this work, we refer to quantum circuits on qubits defined using the Jordan-Wigner mapping as fermionic circuits. We will also refer to unitaries defined by Majorana operators acting on fermions as fermionic circuits.
Augmenting fermionic Gaussian circuits with non-Gaussian gates enables universal quantum computation. In the qubit implementation, adding  SWAP, controlled-$Z$, or controlled-phase gates enables universal quantum computation \cite{Brod_2011}. In the fermionic implementation, adding $\exp(i\pi \gamma_1 \gamma_2 \gamma_3 \gamma_4 /4)$ enables universal quantum computation \cite{bravyi}.

This raises the question whether, in both implementatons, unitaries composed of fermionic Gaussian unitaries and a constant number of non-Gaussian gates are efficiently learnable. In Ref.~\cite{mele2024}, it was shown that quantum states produced by such unitaries are efficiently learnable, leaving open whether the unitaries themselves are efficiently learnable. In this work, we solve this problem by providing an efficient algorithm to learn such unitaries in both fermionic and qubit implementations.
Our result may have applications in fermionic quantum devices, such as those proposed in Refs.~\cite{Gonz_lez_Cuadra_2023, ott2024error, schuckert2024fermionqubitfaulttolerantquantumcomputing} and also complements the classical simulation algorithm for this class of circuits \cite{Dias_2024, Reardon_Smith_2024, browne}. Moreover, previous work has shown that there is an efficient algorithm to learn fermionic unitaries in any finite level of the matchgate hierarchy \cite{cudby2024learninggaussianoperationsmatchgate}. However, we show that fermionic unitaries composed of just two non-Gaussian gates and a Gaussian gate, where the non-Gaussian gates belong to the third level and
the Gaussian gate belongs to the second level \cite{bampounis2024matchgatehierarchycliffordlikehierarchy}, do not belong to any finite level of the matchgate hierarchy.

Our learning algorithm is based on the decomposition result in Ref.~\cite{mele2024} that shows that the unknown unitary acts like a Gaussian unitary on all but a constant number of Majorana operators. 
We show that such Majorana operators can be efficiently learned by measuring expectation values of constant-weight fermionic observables on states prepared using the unknown unitary (e.g.,~using shadow tomography \cite{JuanHam, Zhao_2021}). This information can be used to form a circuit that acts trivially on a large number of Majorana operators. Finally, we construct and learn a circuit that acts on a constant number of modes or qubits either through brute force \cite{Chuang_1997} or by estimating the expectation values of Pauli observables on states prepared by the circuit using shadow tomography~\cite{huangshad, shadowqpt}.

The remainder of the paper is organized as follows. In Sec.~\ref{reviewA}, we define fermionic unitaries, their relation to matchgates, and the computational complexity of such circuits. In Sec.~\ref{mainressec}, we define the learning problem and present the main result. 
In Sec.~\ref{AlgSteps}, we present the technical lemmas on which the learning algorithm is based. In Sec.~\ref{mchier}, we study how the fermionic unitaries in this work relate to unitaries in the matchgate hierarchy and show that, generally, fermionic unitaries with just two non-Gaussian unitaries do not belong to any finite level of the matchgate hierarchy.
In the Appendices, we present details omitted from the main text.

\section{REVIEW: Fermionic unitaries and matchgates}\label{reviewA}

In this section, we present a review of fermionic unitaries, matchgates, and known results regarding their computational complexity. For a fermionic system, the modes labeled $1, \hdots, n$ are defined by the creation and annihilation operators $a_i$ and $a_i^\dagger$, respectively, acting on the Fock-space state $\ket{z_1, \cdots, z_n}$ as follows \cite{bravyi}:
\begin{align}
a_j &\ket{z_1, \hdots,z_{j-1} ,1, z_{j+1}, \hdots,  z_n} \notag\\= &(-1)^{\sum_{k=1}^{j-1}z_k} \ket{z_1, \hdots,z_{j-1} ,0, z_{j+1}, \hdots,  z_n},\\
a_j &\ket{z_1, \hdots,z_{j-1} ,0, z_{j+1}, \hdots,  z_n} = 0,
\end{align}
where $z_i=0,1$. These operators satisfy the anticommutation relations $\{a_i, a_j \}=0$ and $\{a_i, a_j^\dagger \}=\delta_{ij}I$ for all $i,j$. We can define $2n$ Majorana operators $\gamma_i$ as follows: 
\begin{align}
\gamma_{2i-1}&=a_i+a^\dagger_i,\\
\gamma_{2i}&=-i(a_i-a_i^\dagger),
\end{align}
where $i \in [n]$. Here $\gamma_i$ obey the anticommutation relation $\{\gamma_i, \gamma_j \}=0$ for $i\neq j$, and $\gamma_i^2=I$. The $n$-mode Fock space can be associated with the $n$-qubit Hilbert Space \cite{bravyi}. Moreover, the system of $n$ qubits can be mapped to $2n$ Majorana operators $\gamma_i$ ($i = 1, \dots, 2n$) via the Jordan-Wigner transformation defined as
\begin{align}
\gamma_{2k-1}&=\Bigg(\prod_{i=1}^{k-1}Z_i \Bigg) X_k, \label{jw1}\\
\gamma_{2k}&=\Bigg(\prod_{i=1}^{k-1}Z_i \Bigg) Y_k,\label{jw2}
\end{align} 
Here $X_k$, $Y_k$, and $Z_k$ are Pauli matrices acting on qubit $k$, and this mapping ensures that the Majorana operators satisfy the correct anticommutation relations. In this mapping, the Fock-space state
\begin{align}
\ket{z_1 z_2 \hdots z_n}
\end{align}
of $n$ fermionic modes is exactly  identified with the computational basis state on $n$ qubits, where an empty fermionic mode ($z_i = 0$) corresponds to a qubit in computational basis state $\ket{z_i = 0}$, while a filled fermionic mode ($z_i = 1$) corresponds to a qubit in computational basis state $\ket{z_i = 1}$ \cite{bravyi}. 

We define a fermionic Gaussian unitary $G$ from its action on $\gamma_i$ as follows:
\begin{align}
G^\dagger \gamma_i G &= \sum_{k}O_{ik}\gamma_k, \label{canonG}
\end{align}
where $O \in $ O$(2n)$.
Matchgates are parity-preserving two-qubit unitaries of the form
\begin{align}
G(A, B)&=\begin{pmatrix}
A_{11} &0&0&A_{12}\\
0 & B_{11} & B_{12} & 0\\
0 & B_{21} & B_{22}& 0\\
A_{21} & 0& 0 & A_{21}
\end{pmatrix}, \\ \det(A)&=\det(B)=\pm1,
\end{align}
where the matrix is written in the $\{\ket{00},\ket{01},\ket{10},\ket{11}\}$ basis, and where $A$ and $B$ are complex unitary matrices.
To relate matchgates and fermionic Gaussian unitaries, we first note that nearest-neighbor matchgate circuits (on a one-dimensional chain of qubits) can be generated from nearest-neighbor two-qubit gates of the form $\exp(i\theta  X\otimes X)$ and single-qubit $Z$ rotations $\exp(i\theta Z \otimes I)$, $\exp(i\theta I \otimes Z)$. This definition shows that fermionic Gaussian unitaries generated by $G(\theta)_{i i+1}$, where $G(\theta)_{i j}=\exp(\theta \gamma_i \gamma_j)$ with $i\neq j$, and nearest-neighbor matchgate circuits are equivalent since $iX_j X_{j+1} = \gamma_{2j}\gamma_{2j+1}$ and $i Z_j = \gamma_{2j-1}\gamma_{2j}$. Here, unitaries $G(\theta)_{ij}$ 
act on $\gamma_k$ as follows:
\begin{align}
G(\theta)_{ij}^\dagger \gamma_k G(\theta)_{ij} &= \cos (2\theta) \gamma_i + \sin (2\theta) \gamma_j  \>\>\> k=i,\\
G(\theta)_{ij}^\dagger \gamma_k G(\theta)_{ij} &= -\sin (2\theta) \gamma_i + \cos (2\theta) \gamma_j \>\>\> k=j,\\
G(\theta)_{ij}^\dagger \gamma_k G(\theta)_{ij} &= \gamma_k \>\>\> k \notin \{i,j\}.
\end{align}
This shows that the nearest-neighbor $X\otimes X$ rotations and single-qubit $Z$ rotations mentioned earlier can be implemented as fermionic Gaussian unitaries $G(\theta)_{i i+1}$ corresponding to Givens rotations spanned by $\gamma_i$, $\gamma_{i+1}$ for $i \in [2n-1]$, generating all rotations in SO$(2n)$. Adding the operator $X_1$ can then be used to extend to rotations in O$(2n)$ \cite{Wan_2023}. From here on, we will use fermionic Gaussian unitaries and nearest-neighbor matchgate circuits interchangeably. 

Valiant showed that quantum circuits composed of nearest-neighbor matchgates on a one-dimensional chain of qubits can be classically simulated \cite{Valiant}. However, if these gates are allowed to act on arbitrary qubit pairs (or equivalently if the SWAP gate is added to the gate set), such circuits can perform universal quantum computation \cite{kempe}. This result was strengthened in Ref.~\cite{Jozsa_2008} to show that circuits with nearest-neighbor and next-nearest-neighbor matchgate circuits are quantum universal. Moreover, adding any one of the following gates to the gate set
also gives quantum universality: controlled-$Z$, controlled-phase \cite{Brod_2011}, or $\exp(i\pi Z_i Z_j/4)$ \cite{bravyi}. These gates can be written as non-Gaussian unitaries. For example, the gate $\exp(i\pi Z_i Z_j/4)$ can be written as $\exp(i \pi \gamma_{2i-1}\gamma_{2i}\gamma_{2j-1}\gamma_{2j}/4)$. 

We now define the metric used to quantify the precision of our learning algorithm. The diamond-norm distance between any two quantum channels $\mc{E}_1$ and $\mc{E}_2$ can be defined as follows:\begin{align}
\mc D_{\diamond}(\mc E_1, \mc E_2)=\frac{1}{2}\max_{\rho} \norma{(\mc E_1 \otimes I)\rho - (\mc E_2 \otimes I)\rho}_1, \label{defdiamond}
\end{align}
where $\rho$ is a density matrix on 2$n$ qubits, and $\lVert .\rVert_1$ is the trace norm. Moreover, we denote the spectral norm (the largest singular value) as $\norma{.}$ and denote the Frobenius norm as $\norm{.}_2$, which is defined as follows:
\begin{align}
\norma{A}_2 = \sqrt{\tr[A^\dagger A]}.
\end{align}
We use $\alpha_x$ to denote the Hamming weight of the bit-string $x$. The projection on state $\ket{z}$ for the $i$th qubit is denoted by $\Pi_z^{(i)}=(1+(-1)^z Z_i)/2$. We denote single-qubit Paulis acting on the qubit labeled $i$ as $P_i$, where $P \in \{X, Y, Z \}$.
\section{RESULT}\label{mainressec}
We now define our learning problem as follows. We are given black-box access to the unitary $U_t$ with the following promise on its form:
\begin{align}
U_t=G_t K_t \cdots G_1 K_1 G_0, \label{promiseUt}
\end{align}
where $G_i$ is a Gaussian unitary defined in Eq.~(\ref{canonG}), $K_i$ is a non-Gaussian unitary generated by an even-weight product of Majorana operators with weight up to $\kappa$, and $t$ is a constant. We take $\kappa$ to be a constant because such unitaries (e.g.,~$\kappa = 4$ in Ref.~\cite{bravyi}) suffice to implement universal quantum computation. We consider two cases.  The first case is where $U_t$ is implemented on $n$ fermionic modes.
Here, we specialize to parity-preserving  Gaussian unitaries $G_i$. These unitaries  correspond to orthogonal matrices in SO$(2n)$ instead of O$(2n)$ and can be implemented as time evolution under quadratic fermionic Hamiltonians composed of both hopping and pairing terms \cite{bravyi,terhalferm}.
The learning algorithm in this case must use input states that can be efficiently prepared on a fermionic quantum computer (i.e., states that can be obtained from parity-preserving unitaries). We refer to this setting as the \textit{fermionic implementation}. The second case is where $U_t$ is implemented on $n$ qubits. Here, we consider Gaussian unitaries $G_i$ that correspond to orthogonal matrices in O$(2n)$ and are implemented as matchgates. We refer to this setting as the \textit{qubit implementation}.
We aim to find a description of a unitary channel $\mc U_{t}^{(\ell)}$ such that $\mc D_{\diamond}(\mc U_t^{(\ell)}, \mc{U}_t) \leq \epsilon$, where $\mc U_t$ is the quantum channel corresponding to $U_t$. The main result of this work is presented as follows.
 
\begin{restatable}[]{theorem}{mainres}\label{resdet}
\normalfont
For a unitary $U_t$ promised to have the form in Eq.~(\ref{promiseUt}), there is a learning algorithm that accesses the unitary $O(\text{poly}(n, \epsilon^{-1}, \log\delta^{-1}))$ number of times and uses $O(\text{poly}(n, \epsilon^{-1}, \log\delta^{-1}))$ classical processing time to produce a description of the quantum channel $\mc U_{t}^{(\ell)}$ such that $\mc D_{\diamond}(\mc U_t^{(\ell)}, \mc{U}_t) \leq \epsilon$, where $\mc U_t$ is the quantum channel corresponding to $U_t$, with probability $\geq 1-\delta$. The input states used in the algorithm have $O(\text{poly}(n))$ gate complexity. Moreover, the channel description can be used to approximately implement the unitary $U_t$ using $2m+1$ ancillary modes for the fermionic implementation or $2m$ ancillary qubits for the qubit implementation, where $m=M/2$ and $M=\kappa t$. 
\end{restatable}

\noindent Details of the learning guarantees are provided in Lemma \ref{guaranteeAf} for the fermionic implementation and Lemma \ref{guaranteeA} for the qubit implementation. We remark that, since the resource requirements scale as $O(\text{poly}(\log\delta^{-1}))$ in the failure probability $\delta$, we can choose an exponentially small failure probability and still have an efficient algorithm. We remark that we need additional ancilla modes (qubits) to implement the learned unitary because our algorithm involves learning a constant-sized quantum channel that is only approximately unitary, and therefore needs additional ancilla qubits to be implemented as a unitary via Stinespring dilation \cite{Nielsen_Chuang_2010}.

\begin{figure}[t]
\centering
		\includegraphics[
		width= 0.8\columnwidth
		]{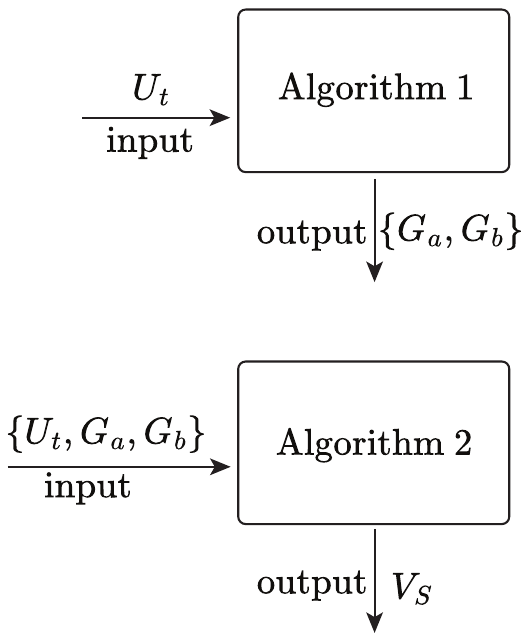}
	\caption{The learning algorithm for the fermionic implementation. Algorithm \ref{alg:algo1} uses access to $U_t$ to produce a description of Gaussian unitaries $G_a$ and $G_b$. Algorithm \ref{alg:algo2} uses access to $U_t$, along with $G_a$ and $G_b$, to learn the unitary $W_t = G_a^\dagger U_t G_b^\dagger$ on a constant number of modes and produce a description of the unitary $V_S$ also suppported on a constant number of modes. The output of the learning algorithm is a description of the unitary $U_t^{(\ell)} = G_a  V_S G_b$. For the qubit implementation, the learning algorithm is the same except that Algorithm \ref{alg:algo2} learns a description of $\bar W_t = \bar U_d^\dagger W_t \bar U_d$ instead of $W_t$.}
	\label{Lflow}
\end{figure}
As shown in Fig.~\ref{Lflow}, our learning algorithm can be described in two steps. In the first step, we perform a tomography scheme that constructs the matrix $ c^{(1)} \in \mathbb{R}^{2n\times 2n}$ defined by 
\begin{align}
c^{(1)}_{jk} := \frac{1}{d}\tr[U_t^\dagger \gamma_k U_t \gamma_j], \label{conedef}
\end{align}
where $d=2^n$. The matrix element $c^{(1)}_{jk}$ can be described as follows. We evolve $\gamma_k$ under $U_t$ and then compute its overlap with the Majorana operator $\gamma_j$.
The reason for the use of the superscript will be made clear later. Since we construct $c^{(1)}$ using tomography, we obtain the matrix $\hat c^{(1)}$ with error $E^{(1)}$ such that $\hat c^{(1)} = c^{(1)} + E^{(1)}$. As shown later in Lemma \ref{constructW}, we can derive Gaussian unitaries $G_a$ and $G_b$ from $\hat c^{(1)}$ such that the unitary $W_t:=G_a^\dagger U_t G_b^\dagger$ satisfies 
\begin{align} [W_t, \gamma_i] \approx 0, \>\>\> i>M \label{approxW2}.
\end{align}
We will refer to this property as Majorana decoupling. The error in this approximation can be made small by performing the tomography step with sufficiently high precision.

We now consider our two implementations. We start with the fermionic implementation, where  we consider Gaussian $G_{j}$ in Eq.~(\ref{promiseUt}) corresponding to orthogonal matrices in SO$(2n)$. In this case, we can ensure $W_t$ is a sum of even-weight Majorana strings by choosing $G_a$ and $G_b$ to be parity-preserving (i.e.\  generated by quadratic fermionic Hamiltonians). We can then show that $[W_t, \gamma_j]=0$ for all $j>M$ implies that $W_t$ is a sum of Majorana strings containing only Majorana operators $\gamma_i$ with $i \leq M$, and therefore $W_t$ only acts on modes $i \in [m]$. To prove this, for some $j > M$, let $A_j$ be the sum of Majorana strings in $W_t$ that contain $\gamma_j$. Since $W_t$ has only even-weight Majorana strings $[W_t, \gamma_j]=0$ implies $[A_j, \gamma_j]=0$. Writing $A_j = B_j \gamma_j$, this immediately implies that $B_j = A_j = 0$. 
In Lemma \ref{lastlemma}, we extend this argument to the case where $[W_t, \gamma_j]=0$ is replaced with $[W_t, \gamma_j] \approx 0$.

We now consider the qubit implementation. Gaussian unitaries $G_{j}$ in Eq.~(\ref{promiseUt}) now correspond to orthogonal matrices in O$(2n)$. In this case, since the Majorana weight of $W_t$ can be odd, the condition in Eq.~(\ref{approxW2}) does not imply that the unitary $W_t$ acts trivially on the qubits labeled $i > m$. As a simple example, take $(n,m)=(2,1)$ and consider $W_t=\gamma_1 \gamma_3 \gamma_4=iX_1Z_2$. We remind the reader that $\gamma_1$ and $\gamma_2$ act on fermionic mode 1, while $\gamma_3$ and  $\gamma_4$ act on fermionic mode  $2$ but on both qubits 1 and 2 due to Jordan-Wigner transformation. Even when $[W_t, \gamma_3]=[W_t, \gamma_4]=0$, $W_t$ acts non-trivially on qubit 2. In this case, we show that we can use a simple unitary transformation $\bar U_d$ to obtain a unitary $\bar W_t=\bar U_d^{\dagger}W_t \bar U_d$ that has no support on qubits labeled $i> m$ when Eq.~(\ref{approxW2}) holds exactly. When Eq.~(\ref{approxW2}) holds approximately, $\bar W_t$ has approximately no support on qubits $i > m$.

We now proceed to the second step of the learning algorithm. Once Algorithm \ref{alg:algo1} allows us to construct unitary transformations of $U_t$ that can be approximated as $m$-mode (qubit) quantum channels for both fermionic and qubit implementations, we can learn those channels. In  Algorithm 2, we learn these channels by measuring the expectation values of Pauli observables that act on the first $m$ modes (qubits) via shadow tomography, and then constructing the corresponding Choi states. The Choi state can be used to compute the unitary Stinespring dilation $V_S$ of the channels \cite{watrous}. The channel can also be learned using brute force \cite{Chuang_1997}. The result of our learning algorithm is a description of the unitary $U_t^{(\ell)}=   G_a   V_S G_b $ acting on $n+2m+1$ modes for the fermionic implementation and a description of the unitary $U_t^{(\ell)}=   G_a   \bar U_d   V_S \bar{U}_d^\dagger G_b $ acting on $n+2m$ qubits for the qubit implementation. The learned unitary satisfies $\mc D_{\diamond}(\mc U_t , \mc U_t^{(\ell)}) \leq \epsilon$, where $\mc{U}_t^{(\ell)}$ is the quantum channel obtained by applying the unitary $U_t^{(\ell)}$ and tracing out the ancillary modes (qubits). As shown in Fig.~{\ref{cptp}}, for the fermionic implementation, the learned unitary can be implemented as a parity-preserving unitary using a unitary transformation on $V_S$, along with an additional ancilla mode \cite{bravyi}.

\begin{figure}[htb]
		\includegraphics[
		width= \columnwidth
		]{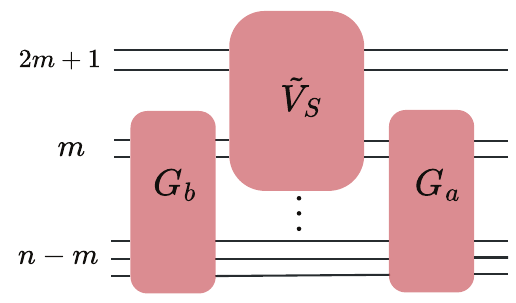}
	\caption{ This figure shows the implementation of the learned quantum channel $\mc U_t^{(\ell)}$ as a product of unitaries acting on $2m+1$ ancillary modes in the fermionic implementation. Here $\tilde V_S$ is a unitary transformation of the Stinespring dilation $V_S$ of the reduced quantum channel $\mc{E}_m^{{\mc{W}}_t}$ corresponding to $W_t$, ensuring that $\tilde V_S$ is parity-preserving.
    }
	\label{cptp}
\end{figure}

\section{METHODS}\label{AlgSteps}

In this section, we first present an overview of the learning algorithm. Then, in Subsecs.~\ref{alg1det} and \ref{alg2det}, we present details and lemmas for Algorithms \ref{alg:algo1} and \ref{alg:algo2}, respectively.

Our learning algorithm is based on a minor modification of the result in Ref.~\cite{mele2024} showing that $U_t$ has the following decomposition.
\begin{restatable}[Decomposition result for $U_t$]{lemma}{Udecomp}\label{Udecomp}
\normalfont
For any given $U_t$ in Eq.~(\ref{promiseUt}), there exist Gaussian unitaries $G_A$ and $G_B$ and unitary $u_t$ such that 
\begin{align}
U_t = G_A u_t G_B, \label{udecompeq}
\end{align}
where $u_t$ is generated by even-weight Majorana strings containing Majorana operators $\gamma_i$ with indices $i \in [M]$. Moreover, for the case where all $G_{j}$ in Eq.~(\ref{promiseUt}) correspond to orthogonal matrices in SO$(2n)$, there exist $G_A$ and $G_B$ that are in SO$(2n)$.
\end{restatable}
\noindent 
See Appendix~\ref{appa} for the proof of this lemma. We can use this result to show that $U_t$ preserves the form of 
all but a constant number of transformed
Majorana operators as follows.
\begin{restatable}[Preserved Majoranas]{lemma}{presmaj}\label{presmaj}
\normalfont
$U_t$ preserves the form of transformed Majorana operators $G_A\gamma_iG_A^\dagger$, for $i=M+1 , \hdots, 2n$, in the following sense:
\begin{align}
U_t^\dagger G_A \gamma_i G_A^\dagger U_t
&=G_B^\dagger \gamma_i G_B \label{eqa}.
\end{align}
Therefore,  
$G_A^\dagger U_t G_B^\dagger$ obeys the Majorana decoupling property
\begin{align}
[G_A^\dagger U_t G_B^\dagger, \gamma_i]=0,\>\>\> i>M.
 \label{freemajWz}
\end{align}
\end{restatable}
\begin{proof} The left-hand side of Eq.~(\ref{eqa}) can be written as follows:
\begin{align}
U_t^\dagger G_A \gamma_i G_A^\dagger U_t &=      G_B^\dagger u_t^\dagger   \gamma_i  u_t G_B\\
&=G_B^\dagger \gamma_i G_B,
\end{align}
where we used the fact that $[u_t, \gamma_i]=0$ from Lemma \ref{Udecomp}. Eq.~(\ref{freemajWz}) follows from Eq.~(\ref{eqa}).
\end{proof}

Lemma \ref{presmaj} shows that Eq.~(\ref{udecompeq}) implies the Majorana decoupling condition in Eq.~(\ref{freemajWz}). For the fermionic implementation, by choosing parity-preserving Gaussian $G_A$ and $G_B$, we also  ensure that Eq.~(\ref{freemajWz}) implies Eq.~(\ref{udecompeq}). 
On the other hand, for the qubit implementation, the existence of Gaussian $G_a$ and $G_b$ such that $W_t=G_a^\dagger U_t G_b^\dagger$ and $[W_t, \gamma_i]=0$ for $i>M$ does not imply that $G_a$ and $G_b$ satisfy Eq.~(\ref{udecompeq}).  This is illustrated for the case $(n,m)=(2,1)$ and the example $U_t =\gamma_4$. Using $G_a=\gamma_1$  and $G_b=-\gamma_3$ gives us $W_t=\gamma_1 \gamma_3 \gamma_4$ which satisfies the Majorana decoupling condition but is not supported on Majorana operators $\gamma_i$ with $i \in [2m]$. This is why, for the qubit implementation, we introduce notation $G_a$ and $G_b$ in addition to $G_A$ and $G_B$. However, we will show that learning $G_a, G_b$ satisfying the  Majorana decoupling condition is sufficient to learn our unknown unitary. Since Lemma \ref{presmaj} shows there exist Gaussian unitaries $G_a$, $G_b$ such that $W_t=G_a^\dagger U_t G_b^\dagger$ satisfies
\begin{align}
[W_t, \gamma_i]=0,\>\>\> i>M, \label{freemajW}
\end{align}
we now proceed to devise a tomographic scheme that discovers these Gaussian unitaries.
Consider the matrix $c_{xk}\in \mathbb{R}^{(T(n)+2n)\times 2n}$ defined as follows:
\begin{align}
c_{xk} = \frac{1}{d}\tr[U_t^\dagger \gamma_k U_t \tilde \gamma_x] \label{defcx},
\end{align}
where $d=2^n$ is the Hilbert space dimension. Here $\tilde \gamma_x$ is defined as
\begin{align}\label{eq:gamma_tilde}
\tilde \gamma_{x}&=\gamma_x \>\>\text{if } \gamma_x^\dagger =\gamma_x,\\
\tilde \gamma_{x}&=i\gamma_x \>\> \text{if }\gamma_x^\dagger =-\gamma_x,
\end{align}
where $\gamma_x$ is a Majorana string defined by the $2n$ bit-string $x$ as $\gamma_x = \gamma_1^{x_1}\hdots \gamma_{2n}^{x_{2n}}$. This definition ensures that our operator basis $\tilde \gamma_x$ is Hermitian. As shown in Fig.~\ref{cxmat}, $c$ is made up of submatrices $c^{(1)}$, defined by Eq.~(\ref{conedef}), and $c^{(2)}$, defined as $c^{(2)}_{xk}=c_{xk}$ with $\alpha_x \geq 2$. 
Moreover, $c^{(1)}$ can be written using its singular-value decomposition as $c^{(1)}=U\Sigma V^T$. As a consequence of Lemma \ref{presmaj}, $c^{(1)}$ has all but a constant number of singular values with value 1. Intuitively, this property is related to Eq.~(\ref{eqa}) which says that there are $2n-M$ Majorana operators $\{G_A \gamma_i G_A^\dagger\}$ that transform under $U_t$ to another set of Majorana operators $\{G_B^\dagger \gamma_i G_B\}$, and the fact that $c^{(1)}$ describes how $U_t$ transforms the Majorana operator $\gamma_i$ to a linear combination of Majorana operators. This is described in the following result (see Appendix \ref{appa} for the proof).

\begin{figure}[htb]
\centering
		\includegraphics[
		width= \columnwidth
		]{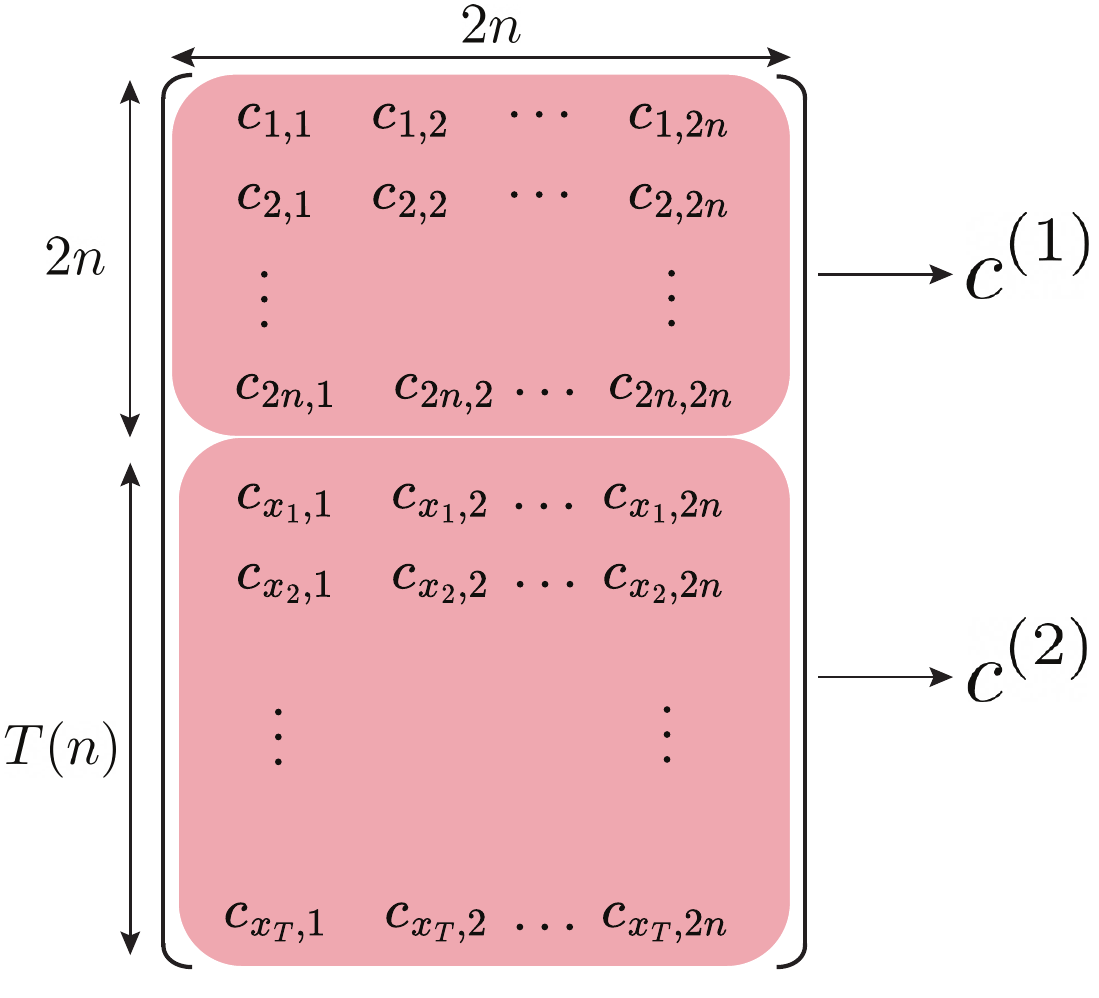}
	\caption{ The matrix $c_{xk}$ defined in Eq.~(\ref{defcx}). Here $T(n)$ is a polynomial defined in Lemma \ref{constructW}.}
	\label{cxmat}
\end{figure}

\begin{restatable}[Useful properties of the matrix $c_{xk}$]{lemma}{propc}\label{propc}
\normalfont
The matrix $c_{xk}$ has orthonormal columns, making all its singular values equal to  1. The matrix $c^{(1)}$ has at least $2n-M$ singular values equal to 1. Moreover, the matrix $  \hat c^{(1)}$, the learned version of $c^{(1)}$, has at least $2n-M$ singular values $\hat d_i$ that satisfy 
\begin{align}
\lvert \hat d_i -1 \rvert \leq \lVert E^{(1)}\rVert,\>\>\> i=M+1, \hdots, 2n \label{sing1},
\end{align}
where $\hat{c}^{(1)}_{jk}=c_{jk}^{(1)}+E_{jk}^{(1)}$, and $E^{(1)}$ is the error from tomography. 
\end{restatable}
We now show that the singular values of $c^{(1)}$ with value 1 correspond to the Majorana operators satisfying Eq.~(\ref{freemajW}). This can be seen from the following computation. Consider the unitary $W_t=G_a^\dagger U_t G_b^\dagger$, where
\begin{align}
G_a \gamma_i G_a^\dagger &= \sum_k O^{a}_{ki}\gamma_k, \label{Gadef}\\
G_b \gamma_i G_b^\dagger &= \sum_k O^b_{ki}\gamma_k.\label{Gbdef}
\end{align}
We can compute the evolved Majorana operator $W_t^\dagger \gamma_i W_t$ with $i >M$ as follows:
\begin{align}
&W_t^\dagger \gamma_i W_t  \notag\\ 
=&\sum_j (c^{(1)}O^a)_{ji}G_b \gamma_j G_b^\dagger + \sum_{x: 2\leq \alpha_x \leq w} (c^{(2)}O^a)_{xi}G_b \gamma_x G_b^\dagger,
\end{align}
where $c^{(2)}_{xi}$ is a submatrix of $c_{xi}$ such that $\alpha_x \geq 2$ ($\alpha_x$ is the Hamming weight of $x$). In the second term, it is sufficient to consider strings with weights upper-bounded by a constant $w=(\kappa+1)^t$. This property follows from the facts that Gaussian unitaries do not increase the Majorana weight of a Majorana string under conjugation, and that $U_t$ only contains a constant number of non-Gaussian unitaries, which can increase the Majorana weight. For more details, see Lemma \ref{majoranawgt} in Appendix \ref{appa}.  
Assuming that the basis is ordered such that $\Sigma_{ii}=1$ for $i>M$, and using $O^a = V$ then gives us
\begin{align}
W_t^\dagger \gamma_i W_t &= \sum_j U_{ji} G_b \gamma_j G_b^\dagger,
\end{align}
where we used the fact that $\Sigma_{ii} = 1$, and that $c_{xk}$ has orthonormal columns. Finally, using Eq.~(\ref{Gbdef}) and $O^b = U^T$ gives us Eq.~(\ref{freemajW}). Since we have access to $\hat c^{(1)}=c^{(1)}+E^{(1)}$ instead of $c^{(1)}$, where $E^{(1)}$ is the tomography error, we have the following result (see Appendix \ref{appa} for the proof).

\begin{restatable}[Constructing the unitary $W_t$]{lemma}{constructW}\label{constructW}
\normalfont
Given $\hat c^{(1)}=c^{(1)} + E^{(1)}$, we can compute descriptions of Gaussian unitaries $G_a$ and $G_b$ such that $W_t=G_a^\dagger U_t G_b^\dagger$ satisfies the following property: 
\begin{align}
\norma{[W_t, \gamma_i]} \leq \epsilon_0 , \>\>\> i>M,\label{approxmaj}
\end{align}
where $\epsilon_0$ obeys the following bound: 
\begin{align}
\epsilon_0&\leq T_1(n) \norma{E^{(1)}}^{1/2},
\label{epsilon0} 
\end{align}
where $T_1(n)$ is a polynomial defined by $T_1(n) = (\sqrt{5}T(n)+2n+1)$ with $T(n) = \sum_{x: 2\leq \alpha_x \leq w} 1=\text{poly}(n)$. Here $x$ is a bitstring of length $2n$, and $w=(\kappa+1)^t$ is a constant defined in Lemma \ref{majoranawgt}.
Here $G_a$ and $G_b$ are defined in Eqs.~(\ref{Gadef}) and (\ref{Gbdef}), where $O^a=V$ and $O^b=U^T$. The orthogonal matrices $U$ and $V$ are defined from the singular-value decomposition of $\hat c^{(1) }$ as $\hat c^{(1)} = U\Sigma V^T$, where $\Sigma = \text{diag}(\hat{d}_1, \hat{d}_2, \hdots, \hat{d}_{2n})$ with $\hat{d}_i$ the singular values  and where $U$, $V$ are orthogonal matrices. Moreover, we can always modify $O^a$ and $O^b$ suitably such that they are in SO$(2n)$ and Eq.~(\ref{approxmaj}) holds.
\end{restatable}
\noindent We therefore have that by learning $\hat{c}^{(1)}$ and performing its singular-value decomposition, we obtain a unitary $W_t=G_a^\dagger U_t G_b^\dagger$ that approximately commutes with the Majorana operators $\gamma_i$ with $i>M$.
In Algorithm \ref{alg:algo1}, we measure observables in states obtained from applications of the unitary $U_t$ to learn the matrix $c^{(1)}$ and use it to compute the descriptions of $G_a$ and $G_b$.
Details of this algorithm are presented in Subsec.~\ref{alg1det}. 

We now proceed to show how the Majorana decoupling condition for $W_t$ from Eq.~(\ref{approxW2}) helps us learn the unitary $U_t$. As discussed in Sec.~\ref{mainressec}, for the fermionic implementation, this condition can be used to show that the unitary $W_t$ acts almost as the identity on modes labeled $i>m$. This is described in the following result (see Appendix \ref{fermionicimp} for the proof).
\begin{restatable}[Majorana decoupling for $W_t$ in the fermionic implementation implies Pauli decoupling for modes $i>m$]{lemma}{lastlemma} \label{lastlemma}
\normalfont
Consider the fermionic implementation where the Gaussian unitaries $G_{j}$ in $U_t$ correspond to orthogonal matrices in SO$(2n)$, and the unitary $W_t$ is obtained from Algorithm \ref{alg:algo1}.
Then $W_t$ satisfies the following:
\begin{align}
\frac{1}{2}\sum_{P \in \{X, Y, Z \}}\norma{[W_t, P_i]} \leq 3n\epsilon_0,\>\>\> i >m+1,\label{lastlemmares}
\end{align}
given $\norma{[W_t, \gamma_j]} \leq \epsilon_0$ for $j>M$ (see Lemma \ref{constructW}).
\end{restatable}

\noindent We will refer to the property obeyed by the unitary in Eq.~(\ref{lastlemmares}) as (approximate) Pauli decoupling. As we will see later, this can be used to show that the unitary acts approximately as the identity on modes labeled $i>m$.

We now consider the qubit implementation where we allow Gaussian unitaries $G_j$ to correspond to orthogonal matrices in O$(2n)$. As discussed in Sec.~\ref{mainressec}, since such unitaries can have odd Majorana weight, the condition in Eq.~(\ref{approxW2}) does not imply that $W_t$ has no support on qubits labeled $i>m$.
Instead, we consider the unitary $\bar W_t$ defined as follows.

\begin{restatable}[ 
]{definition}{defbarw}\label{defbarw}
\normalfont
We define the unitary $\bar W_t$ as follows:
\begin{align}
\bar W_t = \bar U_d^\dagger W_t \bar U_d, \label{defwbar}
\end{align}
where $\bar U_d =  V_d U_d$. The real diagonal unitaries $U_d$  and $V_d$ are defined as
\begin{align}
V_d
&=\sum_x p(\alpha_x) \ketbra{x}{x}_{AB}, \label{vdef2}\\
U_d&=\sum_{x^\prime}p(\alpha_{x^\prime}) \ketbra{x^\prime}{x^\prime}_A.\label{udef2}
\end{align}
Here $U_d$ acts on register $A$ (which contains qubits labeled $i \in [m]$), register $B$ contains qubits labeled $i >m$, $V_d$ acts on all qubits, $\alpha_x$ is the Hamming weight of the state $\ket{x}$, and $p(\alpha)=(-1)^{\alpha (\alpha-1)/2}$.
\end{restatable}
\noindent In the case where $[W_t, \gamma_i]=0$ holds for $i>M$, we can show that $\bar W_t = \bra{0}W_t \ket{0}_A\otimes I_B$ (see Lemma \ref{propW} in Appendix \ref{appb} for more details). We now consider the practical case where $ \norma{[W_t, \gamma_i]} \leq \epsilon_0$  as in Eq.~(\ref{approxmaj}). As shown in Lemma \ref{introwtp} in Appendix \ref{appb}, $\bar W_t$ now satisfies the following:
\begin{align}
\frac{1}{2}\sum_{P \in \{X,  Y, Z \}}\norm{[\bar W_t, P_i]} \leq \epsilon_P, \>\>
\epsilon_P=(2n+3)\epsilon_0,
\label{paulidec}
\end{align}
for all qubits in register $B$.

We now aim to approximate unitaries $W_t$ ($\bar{W}_t$) for the fermionic (qubit) implementation as quantum channels on a constant number of modes (qubits). From this channel, we can then obtain a unitary $V_S$ from the Stinespring dilation applied to the learned quantum channel. To this end, for an arbitrary $n$-mode (qubit) unitary $U$, we introduce the reduced quantum channel $\mc{E}^{\mc{U}}_{m}$ acting on $m$ modes (qubits), defined as follows (see Definition 1 in Ref.~\cite{yunchao} for more details).

\begin{restatable}[Reduced quantum channel \cite{yunchao}]{definition}{locglodef}\label{locglodef}
\normalfont
For a given unitary $U$, we define the reduced channel $\mc{E}^{\mc U}_{m}(\rho)$ that acts on the first $m$ qubits as follows:
\begin{align}
\mc{E}^{\mc U}_{m}(\rho) = \tr_{\geq m+1} \left(U \rho \otimes \frac{I_{B}}{2^{n-m}} U^\dagger \right),
\end{align}
where $\mc U$ is the quantum channel corresponding to the unitary $U$, $I_B$ corresponds to modes (qubits) in register $B$, and  $\tr_{\geq m+1}$ denotes tracing over modes (qubits) $m+1, \hdots , n$ in register $B$.
\end{restatable}
Using ideas developed in Ref.~\cite{yunchao}, Eqs.~(\ref{lastlemmares}) and (\ref{paulidec}) can be used to show that the unitary acts approximately as the identity on modes (qubits) in register $B$. In fact, the reduced quantum channel can be used as a proxy for the unitary channel as shown in the following result.

\begin{restatable}[Approximating a unitary channel  as a reduced quantum channel]{lemma}{localglobal}\label{localglobal}
\normalfont
The channel $\mc{E}^{\mc U}_{m}$ is a CPTP (completely positive and trace preserving) map that satisfies
\begin{align}
\mc{D}_{\diamond}\Big( \mc{U}, \mc{E}^{\mc U}_{m} \otimes \mc I_{B} \Big) \leq n\epsilon  \label{eqcdec},
\end{align}
where $\mc{D}_{\diamond}(\mc{E}_1,\mc{E}_2)$ is defined in Eq.~(\ref{defdiamond}), and $\mc{I}_B$ is the identity channel on register $B$, given the following condition holds:
\begin{align}
\frac{1}{2} \sum_{ P \in \{X, Y, Z  \}}\norm{[U , P_i]} \leq \epsilon
\label{eqadec},
\end{align}
where $i \geq m+1$.
\end{restatable}

\noindent We now proceed to use the above result to approximate the unitaries $W_t$ and $\bar W_t$ in the fermionic and qubit implementations, respectively, as reduced quantum channels. For the fermionic implementation, the property from Eq.~(\ref{lastlemmares}) and Lemma \ref{localglobal} allows us to show that the $m$-mode channel $\mc{E}^{\mc{W}_t}_{m} \otimes \mc I$, where $\mc{E}^{\mc{W}_t}_{m}$ acts on register $A$, is close in diamond distance to the channel corresponding to the unitary $ W_t$ on $n$ modes, as described in Eq.~(\ref{eqcdec}). Similarly, for the qubit implementation, the property from Eq.~(\ref{paulidec}) and Lemma \ref{localglobal} allows us to show that the $m$-qubit quantum channel $\mc{E}^{\bar{\mc{W}}_t}_{m} \otimes \mc I$, where $\mc{E}^{\bar{\mc{W}}_t}_{m}$ acts on register $A$, is close in diamond distance to the channel corresponding to the unitary $\bar W_t$ on $n$ qubits, as described in Eq.~(\ref{eqcdec}).  

We now proceed to learn the $m$-mode (qubit) channels. Any channel $\mc{E}$ on $m$ modes (qubits) can be described via its Choi state
\begin{align}
J(\mc E)=\frac{1}{d_0}\sum_{ij}\mc{E}(\ketbra{i}{j})\otimes \ketbra{i}{j}, \label{choi1}
\end{align}
where $d_0 = 2^m$. 
Since the Choi state $J(\mc {E})$ corresponds to a CPTP map, it also satisfies the following two conditions:
\begin{align}
J(\mc E) &\geq 0, \label{cptp1}\\
\tr_A[J(\mc E)] &=I/d_0 \label{cptp2},
\end{align}
where $A$ denotes the first register in Eq.~(\ref{choi1}).
In Algorithm \ref{alg:algo2}, we use shadow tomography of Pauli observables \cite{yunchao} to learn the Choi state of the channel as $J(\hat{\mc{E}})$ up to some error by measuring
\begin{align}
f_{\alpha \beta}:= \frac{1}{2^n}\tr[S^\dagger (\bar P_\beta\otimes I_{B})S (\bar P_\alpha\otimes I_{B})],\label{paulilearneq}
\end{align}
where $S=W_t$ for the fermionic implementation and $S=\bar{W}_t$ for the qubit implementation, and $\alpha, \beta \in \{I, X, Y, Z \}^{\otimes m}$. 
Note that we use black-box access to $W_t$ or $\bar{W}_t$ when running the tomography process to measure $f_{\alpha \beta}$. As detailed in Section \ref{alg2det}, the coefficients $f_{\alpha\beta}$ can be used to reconstruct the Choi state corresponding to the reduced quantum channels.
Since the state $J(\hat{\mc{E}})$ is learned up to some error, it may not satisfy the CPTP conditions in Eqs.~(\ref{cptp1}) and (\ref{cptp2}). Therefore $J(\hat{\mc{E}})$ is projected (see Subsec.~\ref{alg2det}) to a state $J_p$ which satisfies the CPTP conditions, giving us the following diamond distance bound between the channel $\mc{E}^{\mc{W}_t}$ and the  channel corresponding to the Choi state $J_p$ denoted as $\mc{E}^{\mc{W}_t}_{\text{proj}}$:
\begin{align}
\mc{D}_{\diamond}( \mc{E}^{\mc{W}_t}, \mc{E}^{\mc{W}_t}_{\text{proj}} ) \leq C_3\epsilon_2,
\label{chandisb}
\end{align} 
where $C_3 =   d_0^{11}(3d_0^2+1)/2$ and $\epsilon_2 = \text{max}_{\alpha, \beta}\absa{\hat f_{\alpha \beta} - f_{\alpha \beta}}$. The same result holds for the channel $\mc{E}^{\bar{\mc{W}}_t}$ in the qubit implementation. The Choi state $J_p$ can be used to construct the channel's unitary Stinespring dilation $V_S$ acting on $3m$ modes (qubits) \cite{Nielsen_Chuang_2010}, as shown in Fig.~\ref{cptp}.

This concludes the high-level overview of our learning algorithm. We now describe in detail the main lemmas used to define Algorithms \ref{alg:algo1} and \ref{alg:algo2}.
\subsection{Algorithm \ref{alg:algo1}} \label{alg1det}
We now explain in detail the main ideas and methods used to formulate Algorithm \ref{alg:algo1}. We first consider the qubit implementation. Using methods in Ref.~\cite{JuanHam}, we can estimate $c_{xk}$ by measuring observables $O_k^{+}$ in some state 
prepared by the application of the unknown unitary $U_t$. This leads to the following result (see Appendix \ref{appa} for the proof).
\begin{restatable}[Finding the coefficients $c_{xk}$ (qubit implementation)]{lemma}{lemcx}\label{lemcx}
\normalfont
Let $\mathcal A$ be the $d=2^n$ dimensional Hilbert space upon which the unitary $U_t$ acts. Furthermore, let $\mathcal{B}$ be the Hilbert space of an ancilla register of the same size, and let $\mathcal{C}$ be the Hilbert space of another single ancilla qubit. 
Consider the state $\ket{\psi_x}$ defined as
\begin{align}
\ket{\psi_x} =& \frac{1}{\sqrt{2}}\big[ (U_t\otimes I)\ket{\Phi_d}_{\mathcal{A B}}\ket{0}_{\mathcal C} \notag \\&+ (U_t \otimes I)(\gamma_x^\dagger \otimes I )\ket{\Phi_d}_{\mathcal AB}\ket{1}_{\mathcal C} \big], \label{psixdef}
\end{align}
where
$\ket{\Phi_d}=\frac{1}{\sqrt{d}}\sum_{i=0}^{d-1}\ket{i, i}_{\mathcal{AB}}$
is the maximally entangled state between systems $\mathcal{A}$ and $\mathcal{B}$, and $\gamma_x=\gamma_{1}^{x_1}\hdots \gamma_{ 2n }^{x_{2n}}$. The coefficients $c_{xk}$ [defined in Eq.~(\ref{defcx})] can be obtained using expectation values of observables
\begin{align}
O_k^+ = O_k + O_k^\dagger = (\gamma_k \otimes I)_{\mathcal{AB}}\otimes X_{\mc{C}}, \label{Oplus}\\
O_k^- = iO_k -iO_k^\dagger = (\gamma_k \otimes I)_{\mathcal{AB}}\otimes Y_{\mc{C}} \label{Ominus},
\end{align}
such that $c_{xk}=\tr[\ketbra{\psi_x}{\psi_x} O_k^+]$ for $\gamma_x^\dagger = \gamma_x$, and $c_{xk} = \tr[\ketbra{\psi_x}{\psi_x} O_k^{-}]$ for $\gamma_x^\dagger =-\gamma_x$.
\end{restatable}

\noindent We can apply the above result to measure the matrix elements of $c^{(1)}$ as expectation values of observables $O_k^{+}$ in states $\ket{\psi_j}:=|\psi_{x^{(j)}} \rangle$, where $x^{(j)}$ is a weight-1 bit string with $x_j=1$. The tomography scheme is defined as follows. We first reorder the Hilbert spaces as $\mc C \otimes \mc A \otimes \mc B$ so that observables $O_k^{+}$ can be written as Majorana strings of weight two. This defines new Majorana operators $\hat \gamma_i$ on the Hilbert space  $\mc C \otimes \mc A \otimes \mc B$ as follows:
\begin{align}
\hat{\gamma}_1 &=X_{\mc{C}},\\
\hat{\gamma}_2 &=Y_{\mc{C}},\\
\hat{\gamma}_i&=Z_{\mc{C}} \gamma_{i-2}, \>\>\> i=3, \hdots, 4n+2,
\end{align}
where $\gamma_i$ are the Majorana operators defined in the same way as in Eqs.~(\ref{jw1}) and (\ref{jw2}) on $\mc{A}\otimes \mc{B}$ containing qubits $1, \hdots, 4n$. We can then use the shadow tomography scheme based on the fermionic Gaussian unitary ensemble in Ref.~\cite{Zhao_2021} to obtain estimates of $c^{(1)}_{jk}$, giving us the following result (see Appendix \ref{appa} for the proof).

\begin{restatable}[Estimating the matrix $c^{(1)}$ through shadow tomography for the qubit implementation]{lemma}
{shadowestf}\label{shadowestf}
\normalfont
Using shadow tomography with the fermionic Gaussian unitary ensemble \cite{Zhao_2021}, we can estimate the matrix $c^{(1)}_{jk}=\tr[U_t^\dagger \gamma_k U_t \gamma_j]/d$ by measuring the expectation values of the operators $O_k^{+}$ in state $\ket{\psi_j}$. With probability $\geq 1-\delta$, we obtain the matrix $\hat{c}^{(1)}=c^{(1)}+E^{(1)}$ such that $\norma{E^{(1)}}\leq \norma{E^{(1)}}_2 \leq \epsilon$. For each row $j \in [2n]$ of $c^{(1)}_{jk}$, we need $N_c$ copies of the state $\ket{\psi_j}$, where
\begin{align}
N_c = \left(1+\frac{\epsilon}{6n}\right) \log(8n^2/\delta) \frac{4n^2 (4n+1)}{\epsilon^2}. \label{Mcomp}
\end{align}
Moreover, the required classical post-processing to compute the expectation values can be done efficiently \cite{Zhao_2021}.
\end{restatable}

For the fermionic implementation, the same result holds except that we use 
states
that can be obtained from a parity-preserving quantum circuit, and the observables $O_k^{\pm}$ are modified accordingly. For more details, see see Appendix \ref{fermionicimp}. 

We remark that,  while the shadow tomography step used to construct the matrix $c^{(1)}$ does not flag cases where it produces an inaccurate reconstruction of the matrix, we can make the failure probability $\delta$ of this step  to be exponentially small in $n$ because of the dependence of $N_c$
on $\delta$ in Eqs.~(\ref{Mcomp}).

\begin{algorithm}[h]
\label{alg:algo1}
 \caption{Algorithm for learning the matrix $c^{(1)}$ and descriptions of orthogonal matrices $O^a$ and $O^b$ defining $G_a$ and $G_b$, respectively.}
 \KwIn{Accuracy $\epsilon$, failure probability $\delta$.\\
Qubit implementation: $N_c$ copies of the state $\ket{\psi_j}$ for  each $j \in [2n]$ , where an upper bound on $N_c$ and the definition of $\ket{\psi_j}$ are given in Eqs.~(\ref{Mcomp}) and (\ref{psixdef}), respectively.\\
Fermionic implementation: $N_c^{\text{f}}$ copies of the state $\ket{ \psi_{j}^{\text{f}}}$ for  each $j \in [2n]$ , where $N_c^{\text{f}}$ and $\ket{\psi_j^{\text{f}}}$ are defined in Eqs.~(\ref{expncf}) and (\ref{psijfdef}), respectively.\\
 }

\KwOut{ Matrices $V$ and $U$ obtained from a classical description of the  matrix $\hat c^{(1)}$ such that $\hat{c}^{(1)} = U\Sigma V^T$, and $  \norma{\hat c^{(1)}-c^{(1)}}  \leq \epsilon$ with probability $\geq 1-\delta$.}

For each $j \in [2n]$, perform shadow tomography using the fermionic Gaussian unitary ensemble to estimate Majorana observables $O_k^{+}$ for all $k \in [2n]$, defined in Eq.~(\ref{Oplus}) for the qubit implementation and Eq.~(\ref{oplusferm}) for the fermionic implementation. Construct the matrix $\hat c^{(1)}$ (see Lemma \ref{lemcx} for more details) ;

Compute the singular value decomposition of $\hat c^{(1)} = U \Sigma V^T$, and reorder the basis such that the singular values are written in ascending order;

 \Return  $ V$ and $U$.
\end{algorithm}

\subsection{Algorithm \ref{alg:algo2}} \label{alg2det}
We now explain in detail the main ideas and methods used to formulate Algorithm \ref{alg:algo2}.
The Choi state of any channel $\mc E$ acting on $m$ modes (qubits), defined in Eq.~(\ref{choi1}), can be written as follows:
\begin{align}
J(\mc E)&=\frac{1}{d_0}\sum_{ijkl}\sum_{\alpha \beta} c_{\alpha, ij}c_{\beta, lk}\tr[\mc{E}(\bar P_\alpha)\bar P_\beta] \ketbra{k}{l}\otimes \ketbra{i}{j},
\end{align}
where $d_0=2^m$, $ijkl$ are indices over the computational basis elements of the $m$-mode (qubit) Hilbert space, $\alpha \beta$ are indices used to describe the Pauli string $\bar P_\alpha \in \{I, X, Y, Z \}^{\otimes m}$, and $c_{\alpha, ij}=\tr[\ketbra{i}{j}\bar P_\alpha]/d_0$. As shown in Lemma \ref{noisyChoi} in Appendix \ref{appc}, we can use the result
\begin{align}
\frac{1}{d_0} \tr[\mc{E}(\bar P_\alpha)\bar P_\beta] = f_{\alpha \beta}
\end{align}
to express $J(\mc {E})$ in terms of the matrix elements $f_{\alpha \beta}$, defined in Eq.~(\ref{paulilearneq}). 

We first consider the qubit implementation. We can measure the matrix $f_{\alpha \beta}$ in a similar way to how we measured the matrix $c^{(1)}$, i.e.~by constructing Pauli observables whose expectation values in states obtained from the application of the unknown unitary $U_t$ give $f_{\alpha \beta}$. This is described for the qubit implementation in the following result.

\begin{restatable}[Learning Pauli observables with shadow tomography for the qubit implementation]{lemma}{PauliLearning}\label{PauliLearning} 
\normalfont
The entries of the matrix $f_{\alpha \beta}$ defined as
\begin{align}
f_{\alpha \beta}=\frac{1}{2^n}\tr[S^\dagger (\bar P_\beta\otimes I_{B})S (\bar P_\alpha\otimes I_{B})] \label{paulilearneq2},
\end{align}
where $S=\bar W_t$ from Eq.~(\ref{defwbar}), $\bar P_{\alpha} \in \{I, X, Y, Z \}^{\otimes m}$ are Pauli strings supported on the first $m$ qubits and $\alpha, \beta $ are indices for the set of Pauli strings, can be learned using shadow tomography as follows. We estimate the expectation values of observables
\begin{align}
\bar O_\beta = (\bar{P}_\beta \otimes I)_{\mc{AB}} \otimes \ketbra{1}{0}_{\mc{C}}, 
\end{align}
in states
\begin{align}
\lvert\bar \psi_\alpha \rangle =& \frac{1}{\sqrt{2}}\big[(\bar W_t \otimes I)\ket{\Phi_d}_{\mc{AB}}\ket{0}_{\mc{C}}  \notag\\&+ (\bar W_t \otimes I)_{\mc{AB}}(\bar P_\alpha \otimes I)_{\mc{AB}}\ket{\Phi_d}\ket{1}_{\mc{C}}\big], \label{psialpha}
\end{align}
where $\ket{\Phi_d}$ is the maximally entangled state $\frac{1}{d}\sum_i \ket{i,i}_{\mc{AB}}$, and then construct each row of $\hat f_{\alpha \beta}$ (where $\alpha, \beta \in \{I, X, Y, Z\}^{\otimes m}$) such that $\text{max}_{\alpha, \beta}\absa{\hat f_{\alpha \beta}-f_{\alpha \beta}} \leq \epsilon$ with probability $\geq 1-\delta$. The protocol needs $\bar N_c$ copies of the state $\ket{\bar \psi_\alpha}$, where
\begin{align}
\bar N_c = C_1 \frac{\log(C_2/\delta)}{\epsilon^2}, \label{constlem11}
\end{align}
with $C_1=68(3^m)$, $C_2=2^{2m+1}$.
\end{restatable}

See Appendix \ref{appc} for the proof.
For the fermionic implementation, the same result holds as above except that we use states
that can be prepared by a parity-preserving quantum circuit, and the observables $O_{\beta}^{+}$ are modified accordingly. For more details, see Appendix \ref{fermionicimp}. 

We remark that, while the shadow tomography step used to construct the matrix $f_{\alpha \beta}$ does not flag cases where it produces an inaccurate reconstruction of the matrix, we can make the failure probability of this step $\delta$ to be exponentially small in $n$ because of the dependence of $\bar N_c$
on $\delta$ in Eqs.~(\ref{constlem11}).

Now that we have the learned version of the Choi state $J(\hat{\mc{E}})$ obtained by $\hat f_{\alpha \beta}$ from Algorithm \ref{alg:algo2}, we construct the projected Choi state $J_p$ that satisfies the CPTP conditions in Eqs.~(\ref{cptp1}) and (\ref{cptp2}). The projection scheme is based on ideas in Ref.~\cite{PLS}. The state $J(\hat{\mc{E}})$ is first projected onto a completely positive map denoted by $J_1$. The state $J_1$ is then projected onto a trace-preserving map denoted by $J_2$. Since $J_2$ may have negative eigenvalues, we construct the final state $J_p$ defined by
\begin{align}
J_p = (1-p)J_2 + \frac{p}{d_0^2}\mathbb{1}\otimes \mathbb{1},
\end{align}
 where $p$ is the solution to the equation $(1-p)\lambda_{\text{min}} + p/d_0^2=0$, and $\lambda_{\text{min}}$ is the minimum eigenvalue of $J_2$. One can find this eigenvalue efficiently since the channel has constant dimension. This choice of $p$ ensures $J_p$ has non-negative eigenvalues. We can then show that, given $\norma{J(\hat{\mc{E}})-J(\mc{E})} \leq \epsilon_1$, the projected Choi state $J_p$ obeys $\norma{J(\mc{E}) - J_p}_1 \leq  C_r \epsilon_1$, where $C_r = 3d_0^4+d_0^2$, giving us the channel distance bound in Eq.~(\ref{chandisb}). For more details of the projection scheme, see Lemma \ref{regularizeChoi} in Appendix \ref{appc}. The Choi state $J_p$ on modes (qubits) labeled $1, \hdots, m$, can then be used to construct the unitary Stinespring dilation $V_S$, as shown in Fig.~\ref{cptp}. Using descriptions of $G_a$, $G_b$ (and $\bar U_d$ for the qubit implementation), we obtain the description of the unitary $U_t^{(\ell)}$ that approximates $U_t$. This concludes our exposition of the learning algorithm.

 \begin{algorithm}[h]
\label{alg:algo2}
\caption{Algorithm for learning the Choi state $J(\hat{\mc{E}})$ corresponding to the reduced quantum channel. }
\KwIn{Accuracy $\epsilon$, failure probability $\delta$,\\
Qubit implementation: $\bar N_c $, defined in Eq.~(\ref{constlem11}), copies of the states $\lvert\bar \psi_\alpha \rangle$, defined in Eq.~(\ref{psialpha}), for  each $\alpha \in \{I, X, Y, Z \}^{\otimes m}$.\\
Fermionic implementation: $\bar N_c^{\text{f}}$, defined in Eq.~(\ref{ncfbar}), copies of the states $\lvert\bar \psi_\alpha^{\text{f}} \rangle$, defined in Eq.~(\ref{psialphaf}), for  each $\alpha \in \{I, X, Y, Z \}^{\otimes m}$.\\
}
\KwOut{ A classical description of the state $J( \hat{\mc{E}})$ such that $\max_{\alpha \beta} \lvert \hat f_{\alpha \beta} 
- f_{\alpha \beta} \rvert \leq \epsilon$, and $\norma{J(\hat{\mc{E}}) - J(\mc{E})} \leq d_0^6 \epsilon$, where $J(\mc{E})$ is the Choi state corresponding to the channel $ \mc{E}=  \mc{E}^{{\mc{W}_t}}_m$ for the fermionic implementation, and the channel $ \mc{E}=  \mc{E}_m^{\bar{\mc{W}_t}}$ for the qubit implementation.}

For each $\alpha$, perform shadow tomography using the local Clifford unitary ensemble to estimate Pauli observables $\bar O_\beta^{+}$ defined in Eq.~(\ref{step2obsa}) and Eq.~(\ref{obsf2}) for the qubit implementation and the fermionic implementation, respectively, to construct the matrix $\hat f_{\alpha \beta}$ (see Lemma \ref{PauliLearning} and Appendix \ref{fermionicimp} for more details) ;

Compute the classical description of the learned Choi state $J(\hat{\mc{E}})$ from the matrix elements $\hat f_{\alpha \beta}$;

\Return  $J(\hat{\mc{E}})$\
\end{algorithm}

\section{Matchgate Hierarchy}
\label{mchier}
In this section, we show that fermionic unitaries with a constant number of non-Gaussian gates are, in general, not within the matchgate hierarchy \cite{cudby2023gaussian}. For $n$ qubits (modes), we define an infinite family of gates $\mc{M}_k$ called the matchgate hierarchy. We define the set $\Gamma_1 = \{ \gamma_\mu : \mu \in [2n]\}$. Each set $\mc{M}_k$ can then be defined recursively as follows:
\begin{align}
\mc{M}_1&=\{M \in U(2^n): M=\sum_\mu a_\mu \gamma_\mu, a_\mu \in \mathbb{R} \},\\
\mc{M}_k &= \{ M \in U(2^n) : M \Gamma_1 M^\dagger \subseteq \mc{M}_{k-1}\},\>\>\>k \geq 2. 
\end{align}
We note here that $\mc{M}_2$ corresponds to Gaussian unitaries. 
Recent work  \cite{cudby2023gaussian} has shown that there is an efficient algorithm for learning
unitaries in any finite level of the matchgate hierarchy. We show that arbitrary fermionic unitaries with just two non-Gaussian gates (belonging to the third level of the matchgate hierarchy) lie outside any finite level of the matchgate hierarchy. This is stated in the following lemma.

\begin{restatable}[Example of $U_t$ outside the matchgate hierarchy]{lemma}{outsideMH}\label{outsideMH}
\normalfont
The unitary $U_t=KG(\theta)K$ with two non-Gaussian gates $K$, where 
\begin{align}
K&=\exp(i\pi \gamma_{1}\gamma_{2}\gamma_{3}\gamma_{4}/4),\\
G(\theta)&=\exp(\theta \gamma_1 \gamma_5 ),
\end{align}
$\theta = \pi/p$, and $p$ is an odd integer, does not belong to any finite level of the matchgate hierarchy.
\end{restatable}
\noindent This result can be proved by contradiction. For any unitary $U_t$ to lie in some finite level, say $k$, of the matchgate hierarchy, the unitary $F_1=U_t \gamma_\mu U_t^\dagger$, for any $\mu \in [2n]$, must lie in $\mc{M}_{k-1}$. We can define the unitaries 
\begin{align}
F_k:=F_{k-1}\gamma_\mu F_{k-1}^\dagger, \>\>\>k \geq 2.
\end{align}
Extending the same argument shows that $F_{k-2}$ must lie in $\mc{M}_2$ i.e., $F_{k-2}$ is  Gaussian. Explicit computation for $\mu=2$  shows that $F_{k-2}$ is not Gaussian, showing that $U_t$ does not belong to any finite level of the matchgate hierarchy. See Appendix \ref{matchgateApp} for the proof of this lemma.

\section{CONCLUSION}

In Ref.~\cite{mele2024}, it was shown that there is an efficient algorithm to learn quantum states obtained from fermionic Gaussian unitaries with a constant number of non-Gaussian gates. In this work, we have solved an open problem from Ref.~\cite{mele2024} by providing an efficient algorithm for learning these types of fermionic Gaussian unitaries. We have also shown that such unitaries generally do not fall under the matchgate hierarchy. A few directions for future work could include improving our algorithm by reducing the number of ancilla qubits required to implement the learned unitary or reducing the resource requirements for the learning algorithm. It may be interesting to apply our techniques for learning fermionic channels which in some cases are known to be efficiently simulatable \cite{bravyi2011classicalsimulationdissipativefermionic}. Another possible future direction is that of studying property testing for fermionic unitaries such as those we studied  \cite{bittel2025optimaltracedistanceboundsfreefermionic,lyu2024fermionicgaussiantestingnongaussian}.
 It would also be interesting to explore applications of our learning algorithm to single-parameter and multi-parameter quantum sensing \cite{liu20b}. Our work contributes to a clear understanding about the kinds of quantum processses that can be learned, verified, and benchmarked efficiently, paving the way for accurate and verifiable quantum computation.

\section{Acknowledgments}
S.A.~acknowledges helpful correspondence with Josh Cudby. S.A.~and A.V.G.~acknowledge support from the U.S.~Department of Energy, Office of Science, Accelerated Research in Quantum Computing, Fundamental Algorithmic Research toward Quantum Utility (FAR-Qu). S.A.~and A.V.G.~were also supported in part by the DoE ASCR Quantum Testbed Pathfinder program (awards No.~DE-SC0019040 and No.~DE-SC0024220), NSF QLCI (award No.~OMA-2120757), NSF STAQ program, AFOSR MURI,  DARPA SAVaNT ADVENT, and NQVL:QSTD:Pilot:FTL. S.A.~and A.V.G.~also acknowledge support  from the U.S.~Department of Energy, Office of Science, National Quantum Information Science Research Centers, Quantum Systems Accelerator.  
M.E.S.M.~acknowledges support from the U.S.~Department of Defense through a QuICS Hartree Fellowship. 

\textit{Note added:} While we were polishing the manuscript, a related paper was posted to arXiv \cite{iyer2025mildlyinteractingfermionicunitariesefficiently}. Ref.~\cite{iyer2025mildlyinteractingfermionicunitariesefficiently}  solves the same problem as the one solved by Theorem~\ref{resdet} for the qubit implementation, i.e.\ when the unknown operator has Gaussian unitaries in O$(2n)$. Moreover, in Ref.~\cite{iyer2025mildlyinteractingfermionicunitariesefficiently}, a property testing procedure for testing the closeness of an unknown unitary to unitaries with a constant number of non-Gaussian gates is also considered.

\bibliography{ref}
\let\oldaddcontentsline\addcontentsline
\renewcommand{\addcontentsline}[3]{}
\medskip

\bibliographystyle{apsrev4-2}

\appendix

\section{Details of Algorithm \ref{alg:algo1}}
\label{appa}
In this Appendix, we present proofs for several lemmas related to Algorithm \ref{alg:algo1}. The following lemma describes the decomposition result (Theorem 4 in Ref.~\cite{mele2024}) with minor modifications.
\Udecomp*
\begin{proof}
As in Eq.~(B3) of Ref.~\cite{mele2024}, we can rewrite $U_t$ from Eq.~(\ref{promiseUt}) as $U_t = \tilde{G}_t \prod_{t^\prime =1}^{t} \tilde K_{t^\prime}$, where $\tilde K_{t^\prime} = \tilde{G}_{t^\prime -1}^\dagger K_{t^\prime}\tilde{G}_{t^\prime -1}$ and $\tilde{G}_{t^\prime} = G_{t^\prime}\hdots G_0$. We can then write $U_t$ as follows:
\begin{align}
U_t = \tilde G_t G_{\text{aux}} \prod_{t^\prime =1}^{t}(G_{\text{aux}}^\dagger \tilde K_{t^\prime} G_{\text{aux}})G_{\text{aux}}^\dagger.
\end{align}
This equation holds for an arbitrary $G_{\text{aux}}$.
It is possible to find a $G_{\text{aux}}$ such that the Majoranas that generate each $K_{t^\prime}$ transform under $\tilde G_{t-1}G_{\text{aux}}$ such that they are mapped to the first $M=\kappa t $ Majorana modes. This translates to the condition
\begin{align}
\mathbf{e}_q^T O_{\text{aux}}^T \mathbf{v}_j=0, \label{Oauxeq}
\end{align}
where $q \in \{M+1, \hdots, 2n \}$, $\mathbf{v}_j \in \{ \tilde O_{t^\prime -1}^T \mathbf{e}_{\mu(t^\prime)} \} $,  and $\tilde O_{t^\prime}$ corresponds to the Gaussian $\tilde G_{t^\prime}$ via Eq.~(\ref{canonG}). Here $\mu(t^\prime)$ indexes the Majorana operators in $K_{t^\prime}$ e.g.,~for $K_1=\exp(i\theta\gamma_1 \gamma_2 \gamma_5 \gamma_7)$, we have $\mu(1) \in \{1, 2, 5, 7\}$. Since $K_{t^\prime}$ is generated by a Majorana string with weight $\kappa$, and $t^\prime \in [t]$, we have $j \in [\kappa t]$. We can choose $O_{\text{aux}}$ such that it maps the span of $\{ \mathbf{v}_j \}$ to the first $M=\kappa t$ basis vectors, satisfying Eq.~(\ref{Oauxeq}). Moreover, without loss of generality, we can choose $O_{\text{aux}}$ to be in SO$(2n)$. Concretely, we can define $O_{\text{aux}}^T= \sum_{k=1}^{M} \mathbf{e}_k \mathbf{s}_k^T + \sum_{k=M+1}^{2n} \mathbf{e}_k \mathbf{\bar s}_k^T$, where $\{ \mathbf{s}_i \}$ are orthonormal basis vectors that span $\{\mathbf{v}_j \}$, and $\{ \mathbf{\bar s}_i\}$ are orthonormal vectors outside the span of $\{\mathbf{v}_j \}$. In case $O_{\text{aux}}^T$ is not in SO$(2n)$, we can redefine $O_{\text{aux}}^T \rightarrow O_1 O_{\text{aux}}^T$ where $O_1 = \text{diag}(-1, 1, \hdots, 1)$, ensuring $O_{\text{aux}}$ is in SO$(2n)$ and satisfies Eq.~(\ref{Oauxeq}). We can then write
\begin{align}
U_t = G_A u_t G_B,
\end{align}
where $G_A = \tilde G_t G_{\text{aux}}$, $G_B = G_{\text{aux}}^\dagger$, and $u_t=\prod_{t^\prime}^{t}G_{\text{aux}}^\dagger\tilde{K}_{t^\prime}G_{\text{aux}}$. In the case all $G_{t^\prime}$ correspond to SO$(2n)$, it follows that both $G_A$ and $G_B$ correspond to SO$(2n)$.
\end{proof}

We proceed by establishing a bound on the weight of Majorana strings appearing in $U_t^\dagger \gamma_i U_t$ for $i \in [M]$. This result is necessary to prove the subsequent lemmas.
\begin{lemma}[A bound on the Majorana weight]\label{majoranawgt}
\normalfont
The transformed Majorana operators $\bar{\gamma}_i:= U_t^\dagger \gamma_i U_t$ are sums of  Majorana strings with weight upper bounded by $w:=(\kappa+1)^t$. \label{weightlem}
\end{lemma}
\begin{proof}
First, consider the non-Gaussian unitary $K_l$ as defined in Eq.~(\ref{promiseUt}). Let $K_l$ be generated by $R_l$, a product of $\kappa$ Majorana operators. Since Majorana strings are in the Pauli group, they either commute or anticommute with each other. This means that we need to consider the two cases $R_l = R_l^\dagger$ and $R_l = -R_l^\dagger$.

Let's consider the case $R_l = R_l^\dagger$. We can then write $K_l = e^{-iR_l s}$ (where $s$ is some unknown real parameter). The evolved Majorana operator $ \bar \gamma_i(s):=K_l^\dagger \gamma_i K_l$ is the solution to the differential equation
\begin{align}
\frac{d}{dy}\bar \gamma_i(y) = i e^{iR_l y}[R_l, \gamma_i]e^{-iR_l y}.
\end{align}
In the trivial case wherein $[R_l, \gamma_i]=0$, we have $\bar{\gamma}_i(s)=\gamma_i$,  leaving the Majorana weight unchanged. In the case $\{R_l ,\gamma_i \}=0$, we have
\begin{align}
\frac{d}{dy}\bar{\gamma}_i(y)=2iR_l \bar{\gamma}_i(y) \label{dif1},
\end{align}
where we used the fact that $\{ R_l ,\gamma_i\}=0 \implies [R_l, \gamma_i] = 2R_l \gamma_i$. 
Eq.~(\ref{dif1}) has the solution
\begin{align}
\bar{\gamma}_i(s)=\cos2s \gamma_i + i\sin2s R_l \gamma_i,
\end{align}
where we used the fact that $R_l^2=1$ (from the condition that $R_l=R_l^\dagger$ and $R_l R_l^\dagger =1$).
This shows that $\bar{\gamma}_i(s)$ is a sum of operators with Majorana weight $\leq \kappa+1$.

We now consider the case $R_l = -R_l^\dagger$. We can then write $K_l = e^{R_l s}$. The evolved Majorana operator $\bar{\gamma}_i = K_l^\dagger \gamma_i K_l$ then obeys the differential equation
\begin{align}
\frac{d}{dy}\bar{\gamma}_i(y)=e^{R_l y}[R_l, \gamma_i]e^{-R_l y}.\label{dif2}
\end{align}
In the trivial case wherein $[R_l, \gamma_i]=0$, we have $\bar{\gamma}_i(s)=\gamma_i$,  leaving the Majorana weight unchanged. In the case $\{R_l, \gamma_i \}=0$, Eq.~(\ref{dif2}) then becomes
\begin{align}
\frac{d}{dy}\bar{\gamma}_i(y)=2R_l \bar{\gamma}_i(y),
\end{align}
which has the solution 
\begin{align}
\bar{\gamma}_i(s)=\cos2s \gamma_i + \sin 2s R_l \gamma_i,
\end{align}
where we use the fact that $R_l^2=-1$ (from the condition that $R_l=-R_l^\dagger$ and $R_l R_l^\dagger =1$). This shows that $\bar{\gamma}_i(s)$ is a sum of operators with Majorana weight $\leq \kappa+1$. We therefore have that
\begin{align}
K_l^\dagger \gamma_i K_l  = \sum_{\substack{x|\alpha_x\leq \kappa+1 }} \alpha_x \gamma_x,\label{bound1}
\end{align}
where $x$ is a bit-string of length $2n$, $\gamma_x:=\gamma^{x_1}_1\cdots \gamma^{x_{2n}}_{2n}$, and $\alpha_x$ denotes the Hamming weight of $x$.

Now consider the unitary-evolved operator $ U_t^\dagger \gamma_i U_t$, where $U_t = G_t K_t \cdots G_1 K_1 G_0$. Since evolution under Gaussian unitaries doesn't change the weight of a Majorana operator, Eq.~(\ref{bound1}) shows that the operator $U_t^\dagger \gamma_i U_t$ is a sum of Majorana strings with weight $\leq w:= (\kappa+1)^t$. 
\end{proof}

We now prove Lemma \ref{propc} regarding the singular values of the matrices $c$, $c^{(1)}$, and $\hat{c}^{(1)}=c^{(1)}+E^{(1)}$, where $E^{(1)}$ denotes the tomography error in measuring $c^{(1)}$ using Algorithm \ref{alg:algo1}. 
\propc*
\begin{proof}
We first describe the properties of $\gamma_x$ and $\tilde \gamma_x$, where $\gamma_x = \gamma_1^{x_1}\hdots \gamma_{2n}^{x_{2n}}$, and $\tilde{\gamma}_x$ is defined as
\begin{subequations}\label{defgamx}
\begin{align}
\tilde \gamma_x &= \gamma_x \>\>\>\text{for } \gamma_x^\dagger = \gamma_x,\\
\tilde \gamma_x &= i\gamma_x \>\>\> \text{for }\gamma_x^\dagger = -\gamma_x,
\end{align}
\end{subequations}
making $\tilde \gamma_x$ Hermitian.
Since $\gamma_x$ satisfy the property
\begin{align} \label{tracegx}
\tr[\gamma_x \gamma_y^\dagger]=d \delta_{x,y},
\end{align}
where $d=2^n$, using Eqs.~(\ref{defgamx}) and (\ref{tracegx}), we get
\begin{align}
\tr[\tilde \gamma_x \tilde\gamma_y]=d \delta_{x,y}. \label{trtil}
\end{align}
For $x\neq y$, $\tr[\tilde \gamma_x \tilde \gamma_y]=0$ follows from Eq.~(\ref{tracegx}). The relation $\tr[\tilde \gamma_x^2]=d$ follows from considering the cases $\gamma_x^\dagger =\gamma_x $ and $\gamma_x^\dagger = -\gamma_x$ separately. For $\gamma_x^\dagger =\gamma_x$, $\tr[\tilde \gamma_x^2]=\tr[\gamma_x \gamma_x]=\tr[\gamma_x \gamma_x^\dagger]=d$. For  $\gamma_x^\dagger = -\gamma_x$, $\tr[\tilde \gamma_x^2]=\tr[(i\gamma_x)(i\gamma_x)]=\tr[\gamma_x(-\gamma_x)]=\tr[\gamma_x \gamma_x^\dagger]=d$.

We first show that the columns of $c$ are orthonormal. From the definition of $c_{xk}$, we have
\begin{align}
U_t^\dagger \gamma_k U_t = \sum_{x\in\{0,1\}^{2n}} c_{xk}\tilde{\gamma}_x. \label{gamdecmp}
\end{align}
We then have
\begin{align}
(U_t^\dagger \gamma_j U_t)(U_t^\dagger \gamma_k U_t)&=\sum_{xy}c_{xj}c_{yk}\tilde \gamma_x \tilde \gamma_y,\\
U_t^\dagger \gamma_j \gamma_k U_t &= \sum_{xy}c_{xj}c_{yk}\tilde{\gamma}_{x}\tilde \gamma_{y}. \label{maintr}
\end{align}
Taking $j=k$, and taking the trace of both sides gives us 
\begin{align}
\sum_{x}c_{xj}^2 =1,
\end{align}
where we used the fact that $\tr[\tilde \gamma_x \tilde \gamma_y]=d\delta_{xy}$ from Eq.~(\ref{trtil}). Using $j\neq k$, and taking the trace of both sides of Eq.~(\ref{maintr}) gives us
\begin{align}
\sum_x c_{xj}c_{xk}=0,
\end{align}
where we use Eq.~(\ref{trtil}).
Since $c$ has orthonormal columns, $c^T c=I$, showing that all singular values of $c$ are 1.

We now focus on the singular values of ${c}^{(1)}$. Let $ c^{(2)}$ contain the rows of $ c$ that are not inside $ c^{(1)}$. The matrices $\hat c^{(1)}$, $\hat c^{(2)}$, $E^{(1)}$, and $E^{(2)}$ are defined in the same way,  i.e.~$\hat c^{(1)} = c^{(1)} + E^{(1)}$ and $\hat c^{(2)} = c^{(2)} + E^{(2)}$.  We first define $G_A$ and $G_B$ using orthogonal matrices $O^A$ and $O^B$ as follows:
\begin{align}
G_A \gamma_i G_A^\dagger &= \sum_{k=1}^{2n} O^A_{ki}\gamma_k,\\
G_B \gamma_i G_B^\dagger &= \sum_{k=1}^{2n} O^B_{ki}\gamma_k.
\end{align}
We can expand both sides of Eq.~(\ref{eqa}) in terms of $O^A$, $O^B$, $c^{(1)}$, and $c^{(2)}$ as follows:  
\begin{align}
\sum_j (c^{(1)}O^A)_{ji}\gamma_j +   \sum_{x:\alpha_x\geq 2} (c^{(2)}O^A)_{xi}\tilde{\gamma}_x &= \sum_j O_{ij}^B \gamma_j,
\end{align} 
giving us the equations
\begin{align}
(c^{(1)}O^A)_{ji}&=O^B_{ij}, \label{vec1}\\
(c^{(2)}O^A)_{xi}&=0. \label{vec2}
\end{align}
Now note that the matrix $cO^A$ also has orthonormal columns from the following computation: 
\begin{align}
&(U_t^\dagger G_A \gamma_j G_A^\dagger U_t)(U_t^\dagger G_A \gamma_k G_A^\dagger U_t) \notag\\
&=\sum_{x,y}(cO^A)_{xj}(cO^A)_{yk}\tilde \gamma_x \tilde \gamma_y,\\
&U_t^\dagger G_A \gamma_j \gamma_k G_A^\dagger U_t \notag\\
&=\sum_{xy}(cO^A)_{xj}(cO^A)_{yk}\tilde{\gamma}_x \tilde \gamma_y, \label{trace11}
\end{align}
where we use $U_t^\dagger G_A \gamma_i G_A^\dagger U_t = \sum_k O^A_{ki}U_t^\dagger \gamma_k U_t=\sum_k O_{ki}^A \sum_x c_{xk}\tilde \gamma_x= \sum_x (cO^A)_{xi}\tilde{\gamma}_x$. Considering the case $j=k$ and taking the trace of both sides of Eq.~(\ref{trace11}) and using Eq.~(\ref{trtil}) gives us $\sum_x (cO^A)_{xj}(cO^A)_{xj}=1$. Taking the case $j\neq k$ and taking the trace of both sides of Eq.~(\ref{trace11}) and using Eq.~(\ref{trtil}) gives us $\sum_x (cO^A)_{xj}(cO^A)_{xk}=0$. This means that we can write $cO^A$ as follows: 
\begin{align}
cO^A= \begin{pmatrix}
w_1, \hdots, w_M, v_{M+1}, \hdots, v_{2n}
\end{pmatrix},
\end{align}
where vectors $w_i$ and $v_i$ are real orthonormal column vectors with row-index $x$, where $x$ is a binary string of weight $\leq w$ defined in Lemma \ref{majoranawgt}. Now Eqs.~(\ref{vec1}) and (\ref{vec2}) say that the columns $i=M+1, \hdots, 2n$ of $cO^A$ are orthonormal and are nonzero only on the first $2n$ rows of $cO^A$. We then have
\begin{align}
c^{(1)}O^A= \begin{pmatrix}
\bar w_1, \cdots \bar w_M, \bar v_{M+1}, \cdots \bar v_{2n}
\end{pmatrix},
\end{align}
where $\bar w_i$ and $\bar v_i$ are the truncated versions of $w_i$ and $ v_i$, respectively, such that they contain the first $2n$ elements. Eqs.~(\ref{vec1}) and (\ref{vec2}) say that the vectors $v_i$ are supported on the first $2n$ slots and are orthonormal, giving us the conditions 
\begin{align}
\bar{v}_i^T \bar v_j= v_i^T v_j =\delta_{ij},\>\>\> i, j=M+1, \hdots, 2n.
\end{align}
We now compute $Y^T Y$ with $Y=c^{(1)}O^A$ as follows:
\begin{align}
Y^T Y = \begin{pmatrix}
A &B\\
C&D
\end{pmatrix},
\end{align}
where 
\begin{align}
A_{ij}&=\bar w_{i}^T \bar w_j,\>\>\> i=1, \hdots, M, j=1, \hdots, M,\\
B_{ij}&=\bar w_i^T \bar v_j,\>\>\> i=1, \hdots, M, j=M+1, \hdots, 2n,\\
C_{ij}&= \bar v_i^T \bar w_j, \>\>\> i=M+1, \hdots, 2n, j=1, \hdots, M,\\
D_{ij}&=\bar v_i^T \bar v_j, \>\>\> i=M+1, \hdots, 2n, j=M+1, \hdots, 2n.
\end{align}
Since $v_i$ is only supported on the first $2n$ slots, and $cO^A$ has orthonormal columns, we have $B_{ij}=\bar{w}_i^T \bar{v}_j=w_i^T v_j=0$. Similarly, $C_{ij}=\bar  v_i^T \bar w_j = v_i^T w_j=0$. Finally, $D_{ij}=\bar{v}_i^T \bar{v}_i=\delta_{ij}$, giving us the block diagonal matrix
\begin{align}
Y^T Y=\begin{pmatrix}
A&0\\
0&I
\end{pmatrix},
\end{align}
giving us the result that $c^{(1)}O^A$ has at least $2n-M$ singular values with value 1. Since multiplication by orthogonal matrices doesn't change singular values, it follows that $c^{(1)}$ also has at least $2n-M$ singular values $1$. 

We now consider the case with errors. We first state Weyl's theorem as follows: 
\begin{restatable}[Weyl's theorem \cite{stewart}]{theorem}{weylth}\label{weylth} Let $A$ be a rectangular matrix with singular values $\sigma_1, \hdots, \sigma_n$, and let $\tilde A= A + E$ be a perturbation of $A$ such that $\tilde{A}$ has singular values $\tilde \sigma_1, \hdots \tilde \sigma_n$. Then the following holds:
\begin{align}
\lvert \sigma_i - \tilde \sigma_i\rvert \leq \norma{E}\>\>\> (i=1,2, \hdots, n).
\end{align}
\end{restatable}
\noindent Using the above shows that the matrix $\hat c^{(1)}_{xk}=c_{xk}+E^{(1)}_{xk}$ has at least $2n-M$ singular values that obey Eq.~(\ref{sing1}).

\end{proof}

We now prove Lemma \ref{constructW}, which shows that the unitary $W_t=G_a^\dagger U_t G_b^\dagger$ approximately satisfies the Majorana decoupling condition in Eq.~(\ref{freemajW}).
\constructW*
\begin{proof}
We first order the basis such that for $\hat c^{(1)} = U \Sigma V^T$, $\Sigma e_i =\hat{d}_ie_i$, and $\hat d_i$ satisfies Eq.~(\ref{sing1}) for $i=M+1, \hdots, 2n$. We define $\tilde U_t = G_a^\dagger U_t$, where $G_a$ is defined by the equation $G_a \gamma_i G_a^\dagger = \sum_k O^{a}_{ki}\gamma_k$ and $O^a=V$. We can then compute $\tilde U_t^\dagger \gamma_i \tilde U_t$ as follows:
\begin{align}
\tilde{U}_t^\dagger \gamma_i \tilde{U}_t =&\sum_k O^a_{ki}U_t^\dagger \gamma_k U_t \label{firstlineutil}\\
 =& \sum_k O_{ki}^a   \sum_{x} c_{xk}{\gamma}_x\\
 = &\sum_j (c^{(1)}O^a)_{ji}\gamma_j + \sum_{x: 2\leq \alpha_x\leq w} (c^{(2)}O^a)_{xi}\tilde \gamma_x\\
= & \sum_j (\hat c^{(1)}O^a)_{ji}\gamma_j  - \sum_j (E^{(1)}O^a)_{ji}\gamma_j \notag \\&+ \sum_{x: 2\leq \alpha_x\leq w} (c^{(2)}O^a)_{xi}\tilde \gamma_x \label{first}\\
= &\sum_j U_{ji}\hat d_i \gamma_j -\sum_j (E^{(1)}V)_{ji}\gamma_j \notag\\&+ \sum_{x: 2\leq \alpha_x\leq w} (c^{(2)}V)_{xi}\tilde{\gamma}_x \label{sec}\\
= &\sum_j \hat d_i U_{ji} \gamma_j - \sum_j (E^{(1)}V e_i)_j \gamma_j  \notag \\&+ \sum_{x: 2\leq \alpha_x\leq w} (c^{(2)}V e_i)_{x}\tilde \gamma_x, \label{last1}
\end{align}
where $\alpha_x$ denotes the Hamming weight of $x$, and $w$ is a constant defined in Lemma \ref{majoranawgt}. In Eq.~(\ref{firstlineutil}), we use Eq.~(\ref{gamdecmp}). In Eq.~(\ref{first}), we use $\hat c^{(1)} = c^{(1)}+E^{(1)}$. In Eq.~(\ref{sec}), we use $O^a_{ki}=(V e_i)_k$, where $\Sigma e_i = \hat{d}_i e_i$.  In Eq.~ (\ref{last1}), we use the fact that $V_{ki} = (Ve_i)_k$, where $e_i$ is the $i$th computational basis state.

We now compute $W_t^\dagger\gamma_i W_t$, where $W_t = \tilde U_t G_b^\dagger$ with $G_b \gamma_i G_b^\dagger =\sum_k O^b_{ki}\gamma_k$ and $O^b = U^T$, as follows:
\begin{align}
&W_t^\dagger\gamma_i W_t \notag \\ 
=&\sum_k (O^b U)_{ki}\hat d_i \gamma_k - \sum_k (E^{(1)}V e_i)_{j}G_b \gamma_j G_b^\dagger \notag\\& + \sum_{x:\alpha_x\geq 2} (c^{(2)}V e_i)_{x}G_b \tilde \gamma_x G_b^\dagger \label{useU}\\
=&\hat d_i \gamma_i - \sum_j (E^{(1)}Ve_i)_{j}G_b \gamma_j G_b^\dagger + \sum_{x:\alpha_x\geq 2} (c^{(2)}Ve_i)_{x}G_b \tilde \gamma_x G_b^\dagger.
\end{align}
We can then obtain the bound
\begin{align}
&\lVert  W_t^\dagger \gamma_i W_t - \gamma_i \rVert \notag \\
 \leq& \lvert \hat d_i -1 \rvert + \sum_j \lvert  (E^{(1)}V e_i)_j  \rvert + \sum_{x: 2\leq \alpha_x\leq w} \lvert  (c^{(2)}Ve_i)_x \rvert \\
\leq& \absa{\hat d_i -1} + \absa{E^{(1)}V e_i}\sum_j 1+ \absa{c^{(2)}V e_i}  \sum_{x: 2\leq \alpha_x\leq w} 1\\
\leq & \absa{\hat d_i -1} +2n \absa{E^{(1)}V e_i} + \absa{c^{(2)}Ve_i}T(n)\\
 \leq & \absa{\hat d_i -1} + 2n \norma{E^{(1)}} + \absa{c^{(2)}Ve_i}T(n), \label{almost}
\end{align}
where we use the triangle inequality, the facts that $\lVert \tilde \gamma_x \rVert =1$ and $\lVert .\rVert$ is unitarily invariant. We define $T(n):= \sum_{x:  2\leq\alpha_x\leq w} 1$. We use here the fact that the Majorana weight of $\tilde{\gamma}_x$ in the above equations is bounded by the constant $w$ from Lemma \ref{majoranawgt}. To simplify the last term in the expression above, let us first consider the following:
\begin{align}
\absa{ c V e_i}^2 =& \absa{c^{(1)}V e_i}^2 + \absa{c^{(2)}V e_i}^2\\
=& \absa{\hat c^{(1)} V e_i - E^{(1)}V e_i}^2 + \absa{c^{(2)}V e_i}^2 \label{cvei}\\
=&\absa{\hat c^{(1)}Ve_i}^2 + \absa{E^{(1)}Ve_i}^2 \notag\\&- 2(\hat c^{(1)}V e_i)\cdot (E^{(1)}V e_i) + \absa{c^{(2)}V e_i}^2,
\end{align}
where we use $\hat c^{(1)} = c^{(1)} + E^{(1)}$ in Eq.~(\ref{cvei}), giving us
\begin{align}
\absa{c^{(2)}V e_i}^2 =& \absa{c V e_i}^2 - \absa{\hat c^{(1)}V e_i}^2 - \absa{E^{(1)}V e_i}^2 \notag \\&+2(\hat c^{(1)}V e_i)\cdot (E^{(1)}V e_i) \label{c2Vexp}.
\end{align}
Using the results
\begin{align}
\absa{c Ve_i}^2 &\leq \norma{c}^2 \leq 1, \label{l1}\\
(1-\norma{E^{(1)}})^2 &\leq \absa{\hat c^{(1)}Ve_i}^2 = \hat{d}_i^2 \leq (1+\norma{E^{(1)}})^2, \label{l2}\\
\absa{E^{(1)}Ve_i}^2 &\geq 0,\label{l3}\\
-\hat d_i \norma{E^{(1)}}&\leq (\hat c^{(1)}V e_i)\cdot (E^{(1)}V e_i) \leq \hat d_i \norma{E^{(1)}},\label{l4}
\end{align}
where inequality (\ref{l1}) follows from Lemma \ref{propc}, inequality (\ref{l2}) follows from $\hat c^{(1)}=U\Sigma V^T$ (and the condition that $\norma{E^{(1)}} \leq 1$) and Lemma \ref{propc}, and inequality (\ref{l4}) follows from Cauchy's inequality. Using these inequalities, Eq.~(\ref{c2Vexp}) becomes
\begin{align}
\absa{c^{(2)}V e_i}^2 &\leq 1 -(1-\norma{E^{(1)}})^2 + 2\hat d_i \norma{E^{(1)}} \notag\\
& \leq \norma{E^{(1)}}(2-\norma{E^{(1)}}) + 2\norma{E^{(1)}}(1+\norma{E^{(1)}}) \notag\\
& \leq 5\norma{E^{(1)}},
\end{align}
where we used the condition $\norma{E^{(1)}} \leq 1$. Using the above in inequality (\ref{almost}) gives us
\begin{align}
\norma{[W_t, \gamma_i]}=&\norma{W_t^\dagger\gamma_i W_t -\gamma_i} \notag \\
\leq & (2n+1)\norma{E^{(1)}}  +  (5\norma{E^{(1)}})^{1/2}T(n) \notag\\
\leq & T_1(n) \norma{E^{(1)}}^{1/2},
\end{align}
where $T_1(n)$ is a polynomial defined by $T_1(n) = (\sqrt{5}T(n)+2n+1)$.

We now consider the case where either $U$ or $V$ are outside of SO$(2n)$. The idea is that we can set $G_a \rightarrow G_a (\bar G)^p$ and $G_b \rightarrow (\bar G)^q G_b$, where $p=1$ $(q=1)$ if $U$ ($V$) is outside SO$(2n)$ and $p=0$ $(q=0)$ if $U$ ($V$) is inside SO$(2n)$. We can then consider the unitary $\hat W_t = (\bar G^\dagger)^p W_t (\bar G^\dagger)^q$, where $\bar G$ corresponds to the orthogonal matrix $\bar O = \text{diag}(-1, 1, \hdots, 1)$ using Eq.~(\ref{canonG}). Computing $\hat W_t^\dagger \gamma_i \hat W_t$ for $i>M$ as follows
\begin{align}
\hat W_t^\dagger \gamma_i \hat W_t &=( \bar G)^q W_t^\dagger (\bar G)^p \gamma_i (\bar G^\dagger)^p W_t (\bar G^\dagger)^q\\
&=( \bar G)^q W_t^\dagger \gamma_i W_t (\bar G^\dagger)^q\\
&= (\bar G)^q \gamma_i (\bar G^\dagger)^q+ (\bar G)^q   W_t^\dagger [\gamma_i,  W_t] (\bar G^\dagger)^q\\
&=  \gamma_i+ (\bar G)^q W_t^\dagger [\gamma_i,  W_t] (\bar G^\dagger)^q
\end{align}
gives us $\norma{[\hat W_t, \gamma_i]} \leq \epsilon_0$. Here we used the fact that $\norma{[W_t, \gamma_i]}\leq\epsilon_0$, that $\norma{.}$ is unitarily invariant, and that $\bar G \gamma_i \bar G^\dagger = \gamma_i$ for $i>M$ from the definition of $\bar G$.
\end{proof}

We now prove Lemma \ref{lemcx}, which shows that, in the qubit implementation, the matrix elements $c_{xk}$ can be obtained by measuring observables in states prepared using the unitary $U_t$. This lemma is used in defining Algorithm \ref{alg:algo1} which constructs the matrix $c^{(1)}$, a submatrix of $c_{xk}$. For the fermionic implementation, the analogous result is proved in Appendix \ref{fermionicimp}.     

\lemcx*
\begin{proof}
Let $\rho_x=\ketbra{\psi_x}$.
We show that the operator
\begin{align}
O_k = (\gamma_k \otimes I)_{\mathcal{AB}}\otimes \ketbra{1}{0}_{\mathcal{C}}, \label{okqbit}
\end{align}
can be used to estimate $c_{xk}$ from the following computation. We can write $\rho_x O_k$ as follows: 
\begin{align}
\rho_x O_k =& \frac{1}{2}\big[  (U_t \otimes I) \ketbra{\Phi_d}{\Phi_d} (\gamma_x U_t^\dagger \gamma_k \otimes I) \otimes \ketbra{0}{0}  \notag \\&+(U_t\gamma_x^\dagger \otimes I) \ketbra{\Phi_d}{\Phi_d} (\gamma_x U_t^\dagger \gamma_k \otimes I) \otimes \ketbra{1}{0} \big].
\end{align}
Computing the trace on both sides gives us
\begin{align}
\tr[\rho_x O_k] &= \frac{1}{2}\tr\left[ (U_t\otimes I)\ketbra{\Phi_d}{\Phi_d}(\gamma_x U_t^\dagger \gamma_k \otimes I)\right] \notag\\
&=\frac{1}{2}\tr \left[  (\gamma_x U_t^\dagger \gamma_k U_t \otimes I) \ketbra{\Phi_d}{\Phi_d}\right]\notag\\
&=\frac{1}{2}\bra{\Phi_d}\gamma_x U_t^\dagger \gamma_k U_t \otimes I \ket{\Phi_d}\notag\\
&=\frac{1}{2d}\sum_{i} \bra{i}\gamma_x U_t^\dagger \gamma_k U_t \ket{i}\notag\\
&=\frac{1}{2d}\tr_{\mc A}\left[ U_t^\dagger \gamma_k U_t \gamma_x\right]. \label{expecs1}
\end{align}
Let $a_{xk}=\tr[\rho_x O_k]$. First, consider the case $\gamma_x^\dagger = \gamma_x$. We have
\begin{align}
a_{xk}^* 
&=\frac{1}{2d}\tr[(U_t^\dagger \gamma_k U_t \gamma_x)^{\dagger}]\notag\\
&=\frac{1}{2d}\tr[U_t^\dagger \gamma_k U_t \gamma_x^\dagger]\notag\\
&=a_{xk},
\end{align}
showing $\tr[\rho_xO_k]$ is real. We can compute $\tr[\rho_x O_k^+]$ using Eq.~(\ref{Oplus}) as follows:
\begin{align}
\tr[\rho_x O_k^+] &= \tr[\rho_x O_k]+ \tr[\rho_x O_k^\dagger] \notag\\
&=\tr[\rho_x O_k]+ \tr[(\rho_x O_k)^\dagger] \notag \\
&=2\Re \tr[\rho_x O_k] \notag\\
&=2\tr[\rho_x O_k]\label{lastop},
\end{align}
where we used the fact that $\tr[\rho_x O_k]$ is real. We now consider $\gamma_x^\dagger = -\gamma_x$. We then have
\begin{align}
a_{xk}^* &=\frac{1}{2d}\tr[U_t^\dagger \gamma_k U_t \gamma_x^\dagger]\notag \\
&=-\frac{1}{2d}\tr[U_t^\dagger \gamma_k U_t \gamma_x]\notag\\
&=-a_{xk},
\end{align}
showing $\tr[\rho_x O_k]$ is imaginary. We can compute $\tr[\rho_x O_k^-]$ using Eq.~(\ref{Ominus}) as follows: 
\begin{align}
\tr[\rho_x O_k^-]&=i\tr[\rho_x O_k]-i\tr[\rho_x O_k^\dagger]\notag \\
&=i\tr[\rho_x O_k]-i\tr[(\rho_xO_k)^\dagger]\notag \\
&=i\tr[\rho_x O_k]-i\tr[\rho_x O_k]^*\notag \\
&= -2\Im \tr[\rho_x O_k]\notag\\
&=2i\tr[\rho_x O_k], \label{lastOm}
\end{align}
where we use the fact that $\tr[\rho_x O_k]$ is imaginary. Finally, the coefficient $c_x$ can be written as 
\begin{align}
c_{xk}&=2\tr[\rho_x O_k]=\tr[\rho_x O_k^+] \>\>\> \text{for }\gamma_x^\dagger =\gamma_x,\\
c_{xk}&=2i\tr[\rho_x O_k] = \tr[\rho_x O_k^-]\>\>\> \text{for }\gamma_x^\dagger = -\gamma_x,
\end{align}
where we use the definition of $c_x$ in Eq.~(\ref{defcx}), as well as Eqs.~(\ref{lastop}) and (\ref{lastOm}). We can specialize to the case where $\alpha_x=1$ and $x_j=1$, which makes $c^{(1)}_{jk}$ a real matrix. In this case, we only measure $O_k^+$ since $O_k\propto (O_k^+ - iO_k^-)$ and $c_{jk}^{(1)}\propto \tr[\rho_j O_k]$ (since $c^{(1)}_{jk}$ is a real matrix).
\end{proof}

We now prove Lemma \ref{shadowestf} that gives, for the qubit implementation, guarantees on the error in measuring the matrix $c^{(1)}$ using shadow tomography. The analogous result for the fermionic implementation is proved in Appendix \ref{fermionicimp}.
\shadowestf*
\begin{proof}
We first reorder the Hilbert spaces as $\mc C \otimes \mc A \otimes \mc B$ so that observables $O_k^{+}$ can be written as Majorana strings of weight two. We then define new Majorana operators $\hat{\gamma_i}$
as follows:
\begin{align}
\hat{\gamma}_1 &=X_{\mc{C}},\\
\hat{\gamma}_2 &=Y_{\mc{C}},\\
\hat{\gamma}_i&=Z_{\mc{C}} \gamma_{i-2}, \>\>\> i=3, \hdots, 4n+2,
\end{align}
where $\gamma_i$ are the Majorana operators defined in the same way as in Eqs.~(\ref{jw1}) and (\ref{jw2}) on $\mc{A}\otimes \mc{B}$ containing qubits $1, \hdots, 4n$.
This gives us the following representation of the operators $O_{k}^{\pm}$ defined in Eqs.~(\ref{Oplus}) and (\ref{Ominus}):
\begin{align}
O_k^+&= i\hat{\gamma}_{k+2}\hat{\gamma}_2 \label{okpdef},\\
O_k^-&=-i\hat \gamma_{k+2}\hat \gamma_1.
\end{align}
From  Eq.~(12) in Theorem 2 of Supp.~Mat.~in Ref.~\cite{Zhao_2021}, estimating a Majorana observable $O_j$ with Majorana weight $2k$ with error $\epsilon$ with probability $\geq 1-\delta$ requires $N_c$ copies of the state, where $N_c$ is
\begin{align}
N_c=\left( 1+\frac{\epsilon}{3}\right)\frac{\log(2L/\delta)}{\epsilon^2}\text{max}_{1\leq j \leq L} \norma{O_i}^2_{\mc{U}}, \label{zharores}
\end{align}
where $\norma{O_i}^2_{\mc{U}} =\binom{2\bar n}{2k}/\binom{\bar n}{k}$, where $\bar n$ is the number of qubits, and $2k$ is the Majorana weight of the observable $O_j$ (equal to 2 in our case). Since each state $\ket{\psi_j}$ is used to construct each row of $c^{(1)}_{jk}$, we have $L=2n$ observables. Moreover, because we want the entire error matrix $E^{(1)}$ to have Frobenius norm $\leq \epsilon$ with probability $\geq 1-\delta$, we set $\delta \rightarrow \delta/2n$, and $\epsilon \rightarrow \epsilon/2n$ to obtain Eq.~(\ref{Mcomp}).

\end{proof}

\section{The Pauli decoupling theorem for $\bar W_t$ and guarantee for the reduced quantum channel} \label{appb}

In this Appendix, we first show that $\bar W_t$ from Definition \ref{defbarw} satisfies the Pauli decoupling condition in Eq.~(\ref{paulidec}). We then state and prove Lemma \ref{localglobal}, which shows that the Pauli decoupling property satisfied by $W_t$ in Eq.~(\ref{lastlemmares}) and $\bar W_t$ in Eq.~(\ref{paulidec}) ensures that these unitaries can be approximated by their corresponding reduced quantum channels from Definition \ref{locglodef}. 

Let's first consider the case where $[W_t, \gamma_i]=0$ for $i > M$. We can then show that $\bar W_t$ acts only on the first $m$ qubits. This is shown in the following lemma.

\begin{restatable}[Properties of $\bar W_t$ for exact Majorana decoupling]{lemma}{propW}\label{propW}
\normalfont
In the case where $[W_t, \gamma_i]=0$ for $i=M+1, \hdots, 2n$, the following holds: 
\begin{align}
\bar{W}_t = \bra{\bar{0}}W_t\ket{\bar{0}}_A\otimes {I}_B,
\end{align}
where register $A$ contains qubits labeled $1, \hdots, m$, register $B$ contains qubits labeled $m+1, \hdots, n$, $ I_B $ is the identity on qubits in register $B$, $\ket{\bar{0}}$ is the state $\ket{0^{n-m}}$ defined on qubits in register $B$, and $\bra{\bar{0}}W_t \ket{\bar{0}}$ is an operator defined on register $A$.
\end{restatable}

\begin{proof}
We first consider the following expression for $W_t$:
\begin{align}
W_t = \sum_{\bar{x}, \bar{y}} \bra{\bar{x}}W_t \ket{\bar{y}}_A \otimes \ketbra{\bar{x}}{\bar{y}}_B,
\end{align}
where $\ket{\bar{x}(\bar{y})}$ is a computational basis state on qubits $m+1, \hdots, n$. The first observation is that $\bra{\bar{x}}W_t\ket{\bar{y}}=0$ unless $\bar{x}=\bar{y}$. This is because $\Pi_{\bar{z}}:=I^{[1]}\otimes\ketbra{\bar{z}}{\bar{z}}$, where $I^{[1]}$ is the identity on the qubit block [1] which consists of qubits $1, \hdots, m$, commutes with $W_t$ for all $\bar{z}$. This follows because
\begin{align}
\Pi_{\bar{z}}=\frac{1}{2^{n-m}}[1+(-1)^{z_{m+1}}(-i\gamma_{2m+1} \gamma_{2m+2})] \notag \\ \cdots [1+(-1)^{z_n}(-i\gamma_{2n-1}\gamma_{2n})],
\end{align}
and $[W_t, \gamma_i]=0$ for $i=2m+1, \hdots, 2n$. Then $\bra{\bar{x}}W_t \ket{\bar{y}} = \bra{\bar{x}}W_t \Pi_{\bar{y}}\ket{\bar{y}}=\bra{\bar{x}}\Pi_{\bar{y}}W_t \ket{\bar{y}}=0$ unless $\bar{x}=\bar{y}$. This allows us to simplify $W_t$ as follows:
\begin{align}
W_t = \sum_{\bar{x}} \bra{\bar{x}}W_t \ket{\bar{x}}_A\ketbra{\bar{x}}{\bar{x}}_B. \label{wtdef}
\end{align}
We now relate $\bra{\bar{x}}W_t \ket{\bar{x}}$ to $\bra{\bar{0}}W_t \ket{\bar{0}}$ as follows:
\begin{align}
\bra{\bar{x}}W_t \ket{\bar{x}}&=\sum_{x^\prime y^\prime}\bra{x^\prime \bar{x}}W_t \ket{y^\prime \bar{x}}\ketbra{x^\prime}{y^\prime}\\
&=\sum_{x^\prime y^\prime}\bra{0}\Gamma_{x^\prime \bar{x}}^\dagger W_t \Gamma_{y^\prime \bar{x}}\ket{0}\ketbra{x^\prime}{y^\prime},
\end{align}
where $\ket{0}:=\ket{0_1 \cdots 0_n}$ and
\begin{align}
\Gamma_x:=\gamma_1^{x_1} \gamma_3^{x_2} \hdots \gamma_{2n-1}^{x_n}. \label{Gammadef}
\end{align}
We further simplify the matrix element $\bra{0}\Gamma_{x^\prime \bar{x}}^\dagger W_t \Gamma_{y^\prime \bar{x}}\ket{0}$ as
follows: 
\begin{align}
&\bra{0}\Gamma_{x^\prime \bar{x}}^\dagger W_t \Gamma_{y^\prime \bar{x}}\ket{0} \\ =&  \Gamma_{y^\prime \bar{x}}^2 \Gamma_{x^\prime \bar{x}}^2 \bra{0}\Gamma_{x^\prime \bar{x}}W_t \Gamma_{y^\prime \bar{x}}^\dagger \ket{0}\\
=&\Gamma_{y^\prime \bar{x}}^2 \Gamma_{x^\prime \bar{x}}^2 \bra{0}\Gamma_{x^\prime}\Gamma_{\bar{x}}W_t \Gamma_{\bar{x}}^\dagger \Gamma_{y^\prime}^\dagger
\ket{0}\\
 =&\Gamma_{y^\prime \bar{x}}^2 \Gamma_{x^\prime \bar{x}}^2 \bra{0}\Gamma_{x^\prime}\Gamma_{\bar{x}}\Gamma_{\bar{x}}^\dagger W_t \Gamma_{y^\prime}^\dagger \ket{0}\\
 =& \Gamma_{y^\prime \bar{x}}^2 \Gamma_{x^\prime \bar{x}}^2 \bra{0}\Gamma_{x^\prime}W_t \Gamma_{y^\prime}^\dagger \ket{0}\\
 =& \Gamma_{y^\prime \bar{x}}^2 \Gamma_{x^\prime \bar{x}}^2 \Gamma_{x^\prime}^2 \Gamma_{y^\prime}^2 \bra{0}\Gamma_{x^\prime}^\dagger W_t \Gamma_{y^\prime}\ket{0}\\
 =&\Gamma_{y^\prime \bar{x}}^2 \Gamma_{x^\prime \bar{x}}^2 \Gamma_{x^\prime}^2 \Gamma_{y^\prime}^2 \bra{x^\prime \bar{0}}W_t \ket{y^\prime \bar{0}},
\end{align}
where we use the following facts: 
\begin{align}
\Gamma_x = \Gamma_x^2 \Gamma_x^\dagger,\\
\Gamma_x^2 = (\Gamma_x^2)^\dagger = \pm 1,\\
\Gamma_{x^\prime \bar{y}} = \Gamma_{x^\prime}\Gamma_{\bar{y}},\\
[W_t, \Gamma_{\bar{x}}^\dagger]=0,\\
\Gamma_x \Gamma_x^\dagger =1,
\end{align}
where we define $\Gamma_{x^\prime} = \Gamma_{x^\prime 0 \hdots 0}$  and $\Gamma_{\bar{x}}=\Gamma_{0 \cdots 0 \bar{x}}$, with $x^\prime = x_1, \hdots , x_m$ and $\bar{x} = x_{m+1}, \hdots, x_n$. This gives us
\begin{align}
\bra{\bar{x}}W_t \ket{\bar{x}} = \sum_{x^\prime y^\prime} \Gamma_{y^\prime \bar{x}}^2 \Gamma_{x^\prime \bar{x}}^2 \Gamma_{x^\prime}^2 \Gamma_{y^\prime}^2 \bra{x^\prime \bar{0}}W_t \ket{y^\prime \bar{0}}\ketbra{x^\prime}{y^\prime}.
\end{align}

We can then write
\begin{align}
\bra{\bar{x}}W_t \ket{\bar{x}} = U_d V_{\bar{x}} \bra{\bar{0}}W_t \ket{\bar{0}}V_{\bar{x}}U_d,
\end{align}
where $V_{\bar x}$ and $U_d$ are diagonal unitaries on registers $A$ and $B$, respectively, as follows: 
\begin{align}
V_{\bar{x}}\ket{x^\prime} = \Gamma_{x^\prime \bar{x}}^2 \ket{x^\prime},\\
U_d\ket{x^\prime}=\Gamma_{x^\prime}^2 \ket{x^\prime}.
\end{align}
Using Eq.~(\ref{wtdef}), we then get
\begin{align}
W_t = U_d V_d \left( \bra{\bar{0}}W_t \ket{\bar{0}} \otimes \bar{I} \right)V_d U_d,
\end{align}
where 
\begin{align}
V_d=\sum_{\bar{x}} V_{\bar{x}}\otimes \ketbra{\bar{x}}{\bar{x}}_B.
\end{align}
Moreover, we can simplify the form of the unitary $V$ as follows:
\begin{align}
V_d &= \sum_{\bar{x}} \sum_{x^\prime y^\prime} \bra{x^\prime }V_{\bar{x}}\ket{y^\prime}\ket{x^\prime \bar{x}} \bra{y^\prime \bar{x}}\\
&=\sum_{\bar{x}}\sum_{x^\prime y^\prime} \Gamma_{y^\prime \bar{x}}^2 \braket{x^\prime}{y^\prime} \ketbra{x^\prime \bar{x}}{y^\prime \bar{x}}\\
&=\sum_{x^\prime \bar{x}}\Gamma_{x^\prime \bar{x}}^2 \ketbra{x^\prime \bar{x}}{x^\prime \bar{x}}\\
&=\sum_x \Gamma_x^2 \ketbra{x}{x}.
\end{align}
\end{proof}

We now consider the practical case where $\norm{[W_t, \gamma_i]}=\epsilon_0$ with $i \in [M]$. We first gather a few useful properties about the unitaries $U_d$ and $V_d$ that define $\bar U_d$ in $\bar W_t = \bar U_d^\dagger W_t \bar U_d$ via Eqs.~(\ref{udef2}) and (\ref{vdef2}) in the following lemma.

\begin{lemma}[Some properties of unitaries $U_d$ and $V_d$]
\normalfont
The unitaries $U_d$ and $V_d$ satisfy the following properties.
\begin{enumerate}[label=(\alph*)]
\item The diagonal entries of $U_d$ and $V_d$ satisfy $\bra{x^\prime}U_d\ket{x^\prime}=\Gamma_{x^\prime}^2$ and $\bra{x}V_d \ket{x}=\Gamma_x^2$, respectively, where $x^\prime$ is the computational basis state on qubits $1, \hdots, m$, $x$ is the computational basis state on qubits $1, \hdots, n$, and $\Gamma_x$ is defined in Eq.~(\ref{Gammadef}). Moreover, we can show that $\Gamma_x^2=p(\alpha_x)$, where $\alpha_x$ is the Hamming weight of $x$, and $p(\alpha)=(-1)^{\alpha(\alpha-1)/2}$ obeys the following recursive relation: 
\begin{align}
p(\alpha)=(-1)^{\alpha-1}p(\alpha-1) \label{pdefr},
\end{align}
with $p(0)=1$.
Both $U_d$ and $V_d$ can be efficiently implemented using Hamiltonian simulation of a 2-local Hamiltonian.
\item For $k=1, \hdots, n$,  we have that
\begin{align}
[Z_k, V_d]&=0 \label{eqx2},\\
\Big(\prod_{\substack{l\neq k}}^{n}Z_l \Big) X_k V_d X_k &= V_d. \label{eqx1}
\end{align}

\item For $\norma{[W_t, \gamma_i]} \leq \epsilon_0$, we have the following error bounds: 
\begin{align}
\norm{[W_t, (Z_{m+1}^{x_{m+1}}\cdots Z_n^{x_n})]} \leq 2
\alpha_{y}
\epsilon_0\label{normbound},
\end{align}
where $y=x_{m+1}, \hdots, x_n$  and $\alpha_x$ is the Hamming weight of $x$.
\end{enumerate}
\end{lemma}

\begin{proof}
\item (a) This follows from the definition of the unitaries in Eqs.~(\ref{vdef2}) and (\ref{udef2}). The recursive formula for $\Gamma_x^2$ follows from the definition of $\Gamma_x$.
\item (b) Since $V_d$ and $Z_k$ are diagonal unitaries, they commute.
We now prove Eq.~(\ref{eqx1}). We can write $V_d$ from Eq.~(\ref{vdef2}) as follows: 
\begin{align}
V_d &= \sum_{x\setminus x_k} \Gamma_{\beta_k}^2 \ketbra{\beta_k}{\beta_k} + \Gamma_{\beta_k^\prime}^2 \ketbra{\beta_k^\prime}{\beta_k^\prime},
\end{align}
where $\beta_k =x_1 \hdots x_{k-1}0_k x_{k+1}\hdots x_n$ and $\beta_k^\prime = x_1 \hdots x_{k-1}1_k x_{k+1}\hdots x_n$. First note that
\begin{align}
X_k V_d X_k = \sum_{x\setminus x_k} \Gamma_{\beta_k^\prime}^2 \ketbra{\beta_k}{\beta_k}+ \Gamma_{\beta_k}^2 \ketbra{\beta_k^\prime}{\beta_k^\prime}.\label{xvx}
\end{align}
We can compute the left-hand side of Eq.~(\ref{eqx1}) as follows: 
\begin{align}
\Big(\prod_{l=1, l\neq k}^{n}Z_l \Big) X_k V_d X_k = \sum_{x\setminus x_k}\Gamma_{\beta_k^\prime}^2 (-1)^{\alpha_{k}}\ketbra{\beta_k}{\beta_k} \notag\\ + \Gamma_{\beta_k}^2 (-1)^{\alpha_k} \ketbra{\beta_k^\prime}{\beta_k^\prime},
\end{align}
where we use $\alpha_k:=\alpha_{b}$ with $b=\beta_k$. Finally, using the relation $\Gamma_{\beta_k^\prime}^2 (-1)^{\alpha_k}= \Gamma_{\beta_k}^2$ from Eq.~(\ref{pdefr}), we obtain Eq.~(\ref{eqx1}).

\item(c) Using the commutator identity $[A, BC]=[A, B]C+B[A, C]$, the triangle inequality for the spectral norm, the fact that $\norm{U}=1$ for a unitary $U$, and $Z_i=-i\gamma_{2i-1}\gamma_{2i}$, we get Eq.~(\ref{normbound}).
\end{proof}

We now prove Lemma \ref{introwtp}, which shows that $\bar W_t$ satisfies the Pauli decoupling property in Eq.~(\ref{paulidec}) when $[W_t, \gamma_i] \approx 0$ for $i>M$.

\begin{lemma}[Locality property of $\bar{W}_t$] \label{introwtp}
\normalfont
We consider here the unitary $\bar{W}_t=U_dV_d W_t V_d U_d$, where $W_t$ satisfies $\norm{[W_t, \gamma_i]} \leq \norma{\epsilon_0}$ for $i=2m+1, \hdots, 2n$ from Eq.~(\ref{approxmaj}) in Lemma \ref{constructW}. Then $\bar{W}_t$ satisfies the following properties:
\begin{align}
\norma{[\bar{W}_t, Z_k]} &\leq 2\epsilon_0, \label{eqz}\\
\norma{[\bar{W}_t, X_k]} &\leq (2n+1)\epsilon_0,\label{eqx}\\
\norma{[\bar{W}_t, Y_k]}&\leq 2(n+1)\epsilon_0 \label{eqy},
\end{align}
for $k=m+1, \hdots, n$. Here $\epsilon_0$ is defined in Eq.~(\ref{epsilon0}). Moreover, we can then prove that $\bar W_t$ satisfies the Pauli decoupling property as follows:
\begin{align}
\frac{1}{2}\sum_{P \in \{X,  Y, Z \}}\norm{[\bar W_t, P_i]} \leq \epsilon_P, \>\>
\epsilon_P=(2n+3)\epsilon_0,
\tag{\ref{paulidec}}
\end{align}
where $P_i$ acts on the qubit labeled $i \in \{ m+1, \hdots, n\}$.
\end{lemma}
\begin{proof}
We first prove Eq.~(\ref{eqz}). We have
\begin{align}
\bar{W}_t Z_k &= U_d V_d W_t V_d U_d Z_k\\
&=U_d V_d W_t Z_k V_d U_d\\
&=U_d V_d Z_k W_t V_d U_d + \tilde{O}_3\\
&=Z_k U_d V_d W_t V_d U_d + \tilde{O}_3,
\end{align}
where we used the facts that $U_d$ acts on qubits $1, \hdots, m$, $Z_k$ commutes with $V_d$ from Eq.~(\ref{eqx2}), and that $\tilde{O}_3=U^{(1)} [W_t, Z_k]U^{(2)}$, for some unitaries $U^{(1)}$ and $U^{(2)}$, with $\norm{[W_t, Z_k]}\leq 2\epsilon_0$ from Eq.~(\ref{normbound}). This gives us Eq.~(\ref{eqz}).

We now prove Eq.~(\ref{eqx}). We first compute $X_k W_t X_k$ as follows: 
\begin{align}
X_k W_t X_k &= (Z_1 \cdots Z_{k-1}) \gamma_{2k-1}W_t \gamma_{2k-1}(Z_1\cdots Z_{k-1})\\
&= (Z_1 \cdots Z_{k-1})W_t (Z_1 \cdots Z_{k-1})+O_1\\
&=(Z_1 \cdots Z_{k-1})U_{z,k+1}U_{z, k+1}W_t (Z_1 \cdots Z_{k-1}) \notag \\&+ O_1,
\end{align}
where $O_1=U^{(3)} [W_t, \gamma_{2k-1}]U^{(4)}$ for some unitaries $U^{(3)}$, $U^{(4)}$, and $U_{z, k+1}=\prod_{l=k+1}^{n}Z_l$. Continuing the calculation gives us
\begin{align}
X_k W_t X_k=&(Z_1\cdots Z_{k-1})U_{z, k+1}W_t U_{z, k+1}(Z_1\cdots Z_{k-1}) \notag\\&+O_1+O_2\\
=& \Big(\prod_{l=1, l\neq k}^n Z_l \Big)W_t \Big(\prod_{l=1, l\neq k}^n Z_l \Big) + O_1 + O_2, \label{xwx}
\end{align}
where $O_2=U^{(5)} [U_{z,k+1}, W_t]U^{(6)}$ for some unitaries $U^{(5)}$, $U^{(6)}$. We have $\norm{O_1}\leq \epsilon_0$ from Eq.~(\ref{approxmaj}), and $\norm{O_2}\leq 2n\epsilon_0$ from Eq.~(\ref{normbound}). We can then compute the commutator in Eq.~(\ref{eqx}) as follows. First, consider
\begin{align}
\bar{W}_t X_k &= U_d V_d W_t V_d U_d X_k\\
&=U_d X_k(X_kV_d X_k)(X_k W_t X_k)(X_k V_d X_k)U_d, \label{wxk}
\end{align}
where we insert the identity $X_k^2=1$ in Eq.~(\ref{wxk}) and use the fact that $U_d$ commutes with $V_d$. Inserting the expression for $X_k W_t X_k$ from Eq.~(\ref{xwx}) and gives us
\begin{align}
\bar{W}_t X_k &=  U_d X_k (X_k V_d X_k) \big(\prod_{l=1, l\neq k}^n Z_l \big)W_t \\ &\big(\prod_{l=1, l\neq k}^n Z_l \big) (X_k V_d X_k)U_d + \tilde{O}_1 + \tilde{O}_2 \notag \\
&= X_k U_d V_d W_t V_d U_d + \tilde{O}_1 + \tilde{O}_2, 
\end{align}
where $\tilde{O_{i}}=U^{(i5)}O_i U^{{(i6)}}$ for some unitaries $U^{(i5)}$, $U^{(i6)}$.  We use the fact that $\big(\prod_{l=1, l\neq k}^n Z_l \big) (X_k V_d X_k)  = V_d$ from Eq.~(\ref{eqx1}), and $(X_k V_d X_k) \big(\prod_{l=1, l\neq k}^n Z_l \big)   = V_d$ since both $X_k V_d X_k$ and $\big(\prod_{l=1, l\neq k}^n Z_l \big)$ are diagonal operators (see Eq.~(\ref{xvx}) for the explicit form of $X_k V_d X_k$). This gives us $[\bar{W}_t, X_k]=\tilde{O}_1+\tilde{O}_2$, where $\lvert \lvert \tilde{O}_1 \rvert \rvert \leq \epsilon_0$, $\lvert \lvert {\tilde{O}_2 }\rvert \rvert\leq 2n\epsilon_0$, proving Eq.~(\ref{eqx}).
Finally, using Eqs.~(\ref{eqx}) and (\ref{eqy}) and the fact that $Y=iXZ$ gives us Eq.~(\ref{eqy}). Finally, we can obtain Eq.~(\ref{paulidec}) using Eqs.~(\ref{eqz}--\ref{eqy}).
\end{proof}

We now prove Lemma \ref{localglobal}, which shows that the Pauli decoupling property, shown to hold for $W_t$ and $\bar W_t$ in Eqs.~(\ref{lastlemmares}) and (\ref{paulidec}), respectively,  leads to a reduced quantum channel that is close to the action of the unitary channel. This is proved in the following lemma.

\localglobal*
\begin{proof}
We first show here that $\mc{E}^{U}_{m}$ is a CPTP map. Let $\rho$ be a quantum state over $m$ qubits. The channel $\mc{E}^{U}_{m}$ can be rewritten as follows: 
\begin{align}
\mc{E}^{\mc U}_m(\rho) &= \frac{1}{2^{n-m}}\tr_{\geq m}[U (\rho \otimes I)U^\dagger]\\
&=\frac{1}{2^{n-m}}\sum_{\bar z \bar x} \bra{\bar z}U \ketbra{\bar x}{\bar x}(\rho \otimes I)U^\dagger \ket{\bar z}\\
&=\sum_{\bar z \bar x} E_{\bar z \bar x}\rho E_{\bar z \bar x}^\dagger,
\end{align}
where $E_{\bar z \bar x}=\bra{\bar z}U\ket{\bar x}/\sqrt{2^{n-m}}$ and $\ket{\bar x}=\ket{x_{m+1}\hdots x_n}$ is a state on the qubits $m+1, \hdots, n$. Since $\sum_{\bar z \bar x}E_{\bar z \bar x}^\dagger E_{\bar z \bar x}=I^{\otimes m}$, it follows that $\mc{E}^{U}_{m}$ is a CPTP map (see Corollary 2.27 in \cite{watrous} for details).

We now prove Eq.~(\ref{eqcdec}). Let's first define qubit blocks such that qubits $1, \hdots, m$ are denoted as $[1]$, qubits $m+1, \hdots, n$ are denoted as $[2]$, qubits $n, \hdots, 2n$ are denoted as $[3]$, and qubits $2n+1, \hdots, 3n-m$ are denoted as $[4]$. 

We denote concatenated qubit blocks as $[i,\hdots, j]$, where $[i], [j]$ are the qubit blocks we defined earlier. Consider a state $\rho^{[1,2,3]}$ over registers $[1,2,3]$.  We first note that 
\begin{widetext}
\begin{align}
&\norm{(\mc{U}^{[1,2]} \otimes \mc I^{[3]}) \rho^{[1,2,3]} - (\mc{E}^{\mc U, [1]}_m \otimes \mc I^{[2, 3]}  ) \rho^{[1,2,3]}}_1\\
=&\Big\lVert  \tr_{[4]}\Big[ ( \mc{U}^{[1,2]} \otimes \mc I^{[3,4]})  \left(\rho^{[1,2,3]} \otimes \frac{I^{\otimes n-m}}{2^{n-m}}\right) \Big] - (\mc{E}^{\mc U, [1]}_m \otimes \mc I^{[2, 3]}  ) (\rho^{[1,2,3]}) \Big\rVert_1\\
=&\Big\lVert    \tr_{[2]  }\Big[ (\mc S_{[2], [4]} \circ (\mc{U}^{[1,2]} \otimes \mc I^{[3,4]})) \left(\rho^{[1,2,3]} \otimes \frac{I^{\otimes n-m}}{2^{n-m}}\right) \Big] \notag\\&-(\mc{E}^{\mc U, [1]}_m \otimes \mc I^{ [4,3]} ) (\rho^{[1,4,3]}) \Big\rVert_1 ,
\end{align}
where the swap unitary $S$ between qubit blocks $2$ and $4$ is defined as follows: 
\begin{align}
\mc{S}_{[2],[4]}=\mc S_{m+1, 2n+1} \circ \cdots \circ \mc S_{n, 3n-m},
\end{align}
and the channel corresponding to the swap is denoted by $\mc{S}$. Now we can use the triangle inequality to get

\begin{align}
&\norm{(\mc{U}^{[1,2]} \otimes \mc I^{[3]}) \rho^{[1,2,3]} - (\mc{E}^{\mc U, [1]}_m \otimes \mc I^{[2, 3]}  ) \rho^{[1,2,3]}}_1\\
\leq & \Big\lVert    \tr_{[2]}\Big[ (\mc S_{m+1, 2n+1} \circ \cdots \circ (\mc U^{[1,2]} \otimes \mc I^{[3,4]} ) \circ \mc S_{n, 3n-m} ) \left(\rho^{[1,2,3]} \otimes \frac{I^{\otimes n-m}}{2^{n-m}}\right) \Big] \notag\\&-(\mc{E}^{\mc U, [1]}_m \otimes \mc I^{[4,3]} ) \rho^{[1,4,3]} \Big\rVert_1 + \Big \lVert \mc S_{m+1, 2n+1} \circ \cdots \circ \mc C_1  \left(\rho^{[1,2,3]} \otimes \frac{I^{\otimes n-m}}{2^{n-m}}\right)    \Big \rVert_1, \label{conteq}
\end{align}

\end{widetext}
where $\mc C_1= \mc S_{n,3n-m}\circ (\mc U^{[1,2]} \otimes \mc I^{[3]} )-(\mc U^{[1,2]} \otimes \mc I^{[3]} ) \circ \mc S_{n,3n-m} $, and we used the fact that partial trace cannot increase $\norm{.}_1$. Since all $\mc{S}_{i,j}$ are unitary channels, the last term in Eq.~\eqref{conteq} can be simplified to
\allowdisplaybreaks
\begin{align}
&\big \lVert \mc (S_{m+1, 2n+1} \circ \cdots \circ \mc C_1)  (\rho_1)    \big \rVert_1 \notag\\
=& \big \lVert \mc C_1  (\rho_1)   \big \rVert_1\\
=&\Big\lVert \tr_{[5]}\Big[\mc{C}_1 \otimes \mc{I}^{[5]}  \left(\rho_1 \otimes \frac{I^{\otimes 3n-m}}{2^{3n-m}}\right) \Big] \Big\rVert_1\\
\leq& \Big\lVert  \mc{C}_1 \otimes \mc{I}^{[5]}  \left(\rho_1 \otimes \frac{I^{\otimes 3n-m}}{2^{3n-m}}\right)  \Big\rVert_1\\
\leq& 2\mc{D}_{\diamond}(\mc S_{n, 3n-m}\circ (\mc{U}^{[1,2]}\otimes \mc I^{[3]}), (\mc{U}^{[1,2]}\otimes \mc I^{[3]})\circ \mc S_{n, 3n-m}  ) \label{useyunchao}\\
\leq& 2\epsilon,
\end{align}
where $\rho_1 = \left(\rho \otimes \frac{I^{\otimes n-m}}{2^{n-m}}\right)$, the qubit block $[5]$ is defined by the qubits $3n-m+1, \hdots, 2(3n-m+1)$, and in Eq.~(\ref{useyunchao}), we use the following result from Eq.~(45) in Ref.~\cite{yunchao}: 
\begin{align}
\norm{\mc S_{i,j}(\mc{U} \otimes I_r)-(\mc{U}\otimes I_r)\mc{S}_{i,j} }_{\diamond} \leq 2\epsilon, \label{eqbdec}
\end{align}
where $\mc S_{i,j}$ is the unitary channel corresponding to the swap operator between qubit $i$ in the first $n$ qubits and qubit $j$ in the ancilla register of arbitrary size $r$. To bound the first term in Eq.~\eqref{conteq}, we repeat the same procedure as before with the other swap operators. Repeating the same step for the $j$th time gives the $j$th error term as follows:
\begin{align}
&\Big\lVert \mc S_{m+1, 2n+1}\circ \cdots \circ C_j \circ \cdots \circ S_{n, 3n-m} (\rho_1) \Big \rVert_1 \\=&\Big\lVert \mc S_{m+1, 2n+1}\circ \cdots \circ C_j (\rho_j) \Big \rVert_1\\
=&\Big\lVert C_j (\rho_j) \Big \rVert_1\\
\leq & 2\epsilon, 
\end{align}
where $\rho_j$ is a normalized density matrix since it is $\rho_1$  acted upon some unitary channels. Summing all the error terms then gives us the result from Eq.~(\ref{conteq}) as
\begin{align}
&\norm{(\mc{U}^{[1,2]} \otimes \mc I^{[3]}) \rho^{[1,2,3]} - (\mc{E}^{\mc U, [1]}_m \otimes \mc I^{[2, 3]}  ) \rho^{[1,2,3]}}_1 \notag\\
\leq  & \Big\lVert \tr_{[2]}\Big[ (\mc U^{[1,2]} \otimes \mc I^{[3,4]} ) \circ \mc S_{[2],[4]}   \left(\rho^{[1,2,3]} \otimes \frac{I^{\otimes n-m}}{2^{n-m}}\right)\Big] \notag\\&- (\mc{E}^{\mc U, [1]}_m \otimes \mc I^{[4,3] })(\rho^{[1,4,3]}) \Big\rVert_1 +2(n-m)\epsilon \label{arguef}\\
\leq & 2(n-m)\epsilon.
\end{align}
In line (\ref{arguef}), the first term is zero because $\tr_{[2]}(.)$ is the definition of $(\mc{E}^{\mc U, [1]}_m \otimes \mc I^{[4,3] })\rho^{[1,4,3]}$. Finally, using the definition of $\mc{D}_{\diamond}(.)$ from Eq.~(\ref{defdiamond}) gives us the result in Eq.~(\ref{eqcdec}). 
\end{proof}

\section{Details of Algorithm \ref{alg:algo2}}\label{appc}
In this Appendix, we gather a few technical lemmas to prove the guarantees provided in Algorithm \ref{alg:algo2}. In Appendix \ref{falpbet}, we prove Lemma \ref{PauliLearning} that shows how to measure $f_{\alpha \beta}$ used to construct the Choi state of the quantum channel $\mc{E}^{\bar{\mc{W}}_t}_{m}$ for the qubit implementation. The analogous result for the fermionic implementation is provided in Appendix \ref{fermionicimp}. In Appendix \ref{appdz}, we prove technical lemmas regarding the Choi states corresponding to the reduced quantum channels learned in Algorithm \ref{alg:algo2} for both fermionic and qubit implementations. We also provide the learning guarantee for the qubit implementation (see Appendix \ref{fermionicimp} for the learning guarantee in the fermionic implementation).

\subsection{Learning the matrix $f_{\alpha \beta}$}\label{falpbet}
\PauliLearning*
\begin{proof}
We can compute $\bar \rho_\alpha O_\beta$, where $\bar \rho_{\alpha}=\lvert \bar \psi_\alpha \rangle \langle \bar \psi_\alpha \rvert$, as follows:
\begin{align}
\bar \rho_\alpha \bar O_\beta =&\frac{1}{2} \big[ (\wtp \otimes I)\ketbra{\Phi_d}{\Phi_d}(\bar P_\alpha \wtpd \bar P_\beta \otimes I)\otimes \ketbra{0}{0}_{\mc C} \notag \\&+ (\wtp \bar{P}_\alpha \otimes I)\ketbra{\Phi_d}{\Phi_d}(\bar P_\alpha \wtpd \bar P_\beta \otimes I)\otimes \ketbra{1}{0}_{\mc C} \big].
\end{align}
Taking the trace of the above gives us
\begin{align}
\tr[\bar \rho_\alpha \bar O_\beta] &= \frac{1}{2}\tr[(\wtp \otimes I) \ketbra{\Phi_d}{\Phi_d}(\bar P_\alpha \wtpd \bar P_\beta \otimes I)] \notag\\
&=\frac{1}{2}\tr[(\bar P_\alpha \wtpd \bar P_\beta \wtp \otimes I)\ketbra{\Phi_d}{\Phi_d}]\notag\\
&=\frac{1}{2}\bra{\Phi_d}\bar P_{\alpha}\wtpd \bar P_\beta \wtp \otimes I \ket{\Phi_d}\notag\\
&=\frac{1}{2d}\sum_i \bra{i}\bar P_{\alpha}\wtpd \bar P_\beta \wtp \ket{i}\notag\\
&=\frac{1}{2}\frac{1}{2^n}\tr[\bar P_{\alpha}\wtpd \bar P_\beta \wtp]\notag\\
&=\frac{1}{2}f_{\alpha \beta},
\end{align}
where we use $d=2^n$. We can then write the operator $\bar O_\beta$ as 
\begin{align}
\bar O_\beta = \frac{1}{2}(\bar O_\beta^+ -i\bar O_\beta^-),
\end{align}
where $\bar O_\beta^{\pm}$ are Hermitian operators defined as
\begin{align}
\bar O_\beta^+ &= (\bar P_\beta \otimes I)_{\mc{AB}}\otimes X_{\mc C},\label{step2obsa}\\
\bar O_\beta^- &=(\bar P_\beta \otimes I)_{\mc{AB}} \otimes Y_{\mc C}, \label{step2obsb}
\end{align}
giving us $f_{\alpha \beta}=\tr\small[\bar \rho_\alpha  \bar O_\beta^+ \small] -i\tr \small[ \bar \rho_{\alpha }\bar O_\beta^- \small]$. Since $f_{\alpha \beta}$ is a real matrix, we have $f_{\alpha \beta} = \tr\small[\bar \rho_\alpha \bar O_\beta^+\small]$.  Ref.~\cite{huangshad} on shadow tomography using the local Clifford unitary ensemble shows that estimating the expectation value of tensored single-qubit Paulis acting non-trivially on $k$ qubits in some state, with probabability $\geq 1-\delta$ and error  $\epsilon$, needs $68.3^k \log(2L/\delta)\epsilon^2$ copies of the state. Here $L$ is the number of observables. For additional details, see Eqs.~(S13) and (S50) in Ref.~\cite{huangshad}. Since we want to construct the matrix $\hat f_{\alpha \beta}$ such that $\text{max}_{\alpha \beta}\absa{\hat f_{\alpha \beta}-f_{\alpha \beta}}\leq \epsilon$, we need to estimate observables $\bar O_\beta^{+}$ for each state $\lvert\bar \psi_\alpha \rangle$ with error $\epsilon$ with probability $\geq 1- \delta/4^{m}$ (since the index $\alpha $ has $\leq 4^m$ many values). The number of copies $\bar N_c$ of each state $\lvert\bar \psi_\alpha \rangle$  is
\begin{align}
\bar N_c = \frac{68}{\epsilon^2}3^m \log(2^{2m+1}/\delta),
\end{align}
where we use the facts that there are $ \leq 4^m$ observables for each $\alpha$ and that the observables $\bar O_\beta^{+}$ have Pauli weight $\leq m+1$.
\end{proof}
\subsection{Technical lemmas for Choi states}
\label{appdz}
In this subsection, we state and prove key technical lemmas regarding the Choi state learned in Algorithm \ref{alg:algo2}. The results presented here apply for channels corresponding to the unitaries $W_t$ (for the fermionic implementation) and $\bar W_t$ (for the qubit implementation).

\begin{lemma}[Learning the reduced quantum channel $\mc{E}^{\mc{Q}}_m$ from shadow tomography] \label{noisyChoi} 
\normalfont

We can learn the Choi state of the reduced quantum channel $\mc{E}^{\mc{Q}}_m$, corresponding to the unitary $Q$, such that the distance between the learned Choi state $J(\hat{\mc{E}})$ and the Choi state $J(\mc E)$ corresponding to $\mc{E}^{\mc{Q}}_m$ is bounded as
\begin{align}
\norma{J(\hat{\mc E}) - J(\mc E)} \leq d_0^6\delta, \label{jchoinorm}
\end{align}
where $\delta =   \text{max}_{\alpha \beta} \lvert q_{\alpha, \beta} - \hat q_{\alpha, \beta}\rvert $,  $q_{\alpha, \beta}$ is defined as
\begin{align}
q_{\alpha \beta}:= \frac{1}{2^n}\tr[Q^\dagger (\bar P_\beta\otimes I_{B})Q (\bar P_\alpha\otimes I_{B})],
\end{align}
and $\hat q$ is the learned version of $q$.

\end{lemma}
\begin{proof}
 
We have that
\begin{align}
&d_0 q_{\alpha \beta}\\
=&\frac{1}{2^{n-m}}\tr[Q^\dagger  (\bar P_\beta \otimes I_{B})Q (\bar P_\alpha\otimes I_{B})] \notag \\
=&\frac{1}{2^{n-m}}\tr[Q (\bar P_\alpha \otimes I_{B})Q^\dagger (\bar P_\beta \otimes I_{B})] \notag \\
=&\frac{1}{2^{n-m}}\tr(\tr_{\geq m+1}[Q (\bar P_\alpha\otimes I_{B}) Q^\dagger (\bar P_\beta \otimes I_{B})]) \notag \\
=&\frac{1}{2^{n-m}}\tr(\tr_{\geq m+1}[Q (\bar P_\alpha \otimes I_{B})Q^\dagger]\bar P_{\beta}) \notag \\
=&\tr(\mc{E}^{\mc{Q}}_m(\bar P_\alpha)\bar P_\beta),
\end{align}
giving us the result
\begin{align}
2^m q_{\alpha \beta}=\tr(\mc{E}^{\mc{Q } }_m(\bar P_\alpha)\bar P_\beta).
\end{align}
For any channel $\mc{E}$ on $m$ modes (qubits), we can write the Choi-Jamiolkowski state $J(\mc{E})$ as follows:
\begin{align}
J(\mc E)=\frac{1}{d_0}\sum_{ij}\mc{E}(\ketbra{i}{j})\otimes \ketbra{i}{j}, \label{choi}
\end{align}
where $\ket{i}$ is a computational basis state on $m$ modes (qubits), and $d_0=2^m$. We can expand $J(\mc{E})$ as follows: 
\begin{align}
J(\mc{E})&=\frac{1}{d_0}\sum_{ijkl}\tr[\mc{E}(\ketbra{i}{j})\ketbra{l}{k}]\ketbra{k}{l}\otimes \ketbra{i}{j}\\
&=\frac{1}{d_0}\sum_{ijkl}\sum_{\alpha \beta} c_{\alpha, ij}c_{\beta, lk}\tr[\mc{E}(\bar P_\alpha)\bar P_\beta] \ketbra{k}{l}\otimes \ketbra{i}{j}\\
&=\sum_{ijkl}\sum_{\alpha \beta} c_{\alpha, ij}c_{\beta, lk}q_{\alpha \beta}\ketbra{k}{l}\otimes \ketbra{i}{j}. \label{Je}
\end{align}
We can then bound $\norma{J(\hat {\mc E}) - J(\mc E)}$, where $J(\hat{\mc{E}})$ corresponds to the Choi state constructed using $\hat q_{\alpha \beta}$, as follows: 
\begin{align}
&\norma{J(\hat{\mc E}) - J(\mc E)}  \notag \\  \leq &\sum_{ijkl}\sum_{\alpha \beta}\absa{c_{\alpha, ij}c_{\beta, lk}} \absa{\hat q_{\alpha \beta}-q_{\alpha \beta}}\norma{\ketbra{k}{l}\otimes \ketbra{i}{j}}.\label{Jspeca}
\end{align}
We used the facts that $\norma{\ketbra{k}{l}\otimes \ketbra{i}{j}} \leq 1$ and $c_{\alpha, ij} \leq 1/2^m$ from $c_{\alpha, ij} =\tr[\ketbra{i}{j} \bar P_\alpha]/2^m = \bra{j}\bar P_\alpha \ket{i}/2^m $.  Equation~(\ref{Jspeca}) then becomes
\begin{align}
\norma{J(\hat{\mc E}) - J(\mc E)} &\leq \frac{\delta}{2^{2m}}\sum_{ijkl}\sum_{\alpha \beta}1 \notag\\
&=2^{6m}\delta, \label{specboundjje}
\end{align}
where we use $\sum_{i}1=2^m$ and $\sum_\alpha 1 \leq 4^m$. Taking  $\delta = \text{max}_{\alpha, \beta}\absa{\hat q_{\alpha \beta}-q_{\alpha \beta}}$ then gives us the result in Eq.~(\ref{jchoinorm}). We also get the inequality
\begin{align}
\norma{J(\hat{\mc{E}}) - J(\mc{E})}_1 \leq d_0^8 \delta\label{troneboundjje}
\end{align}
from using the inequality  $\norma{X}_1 \leq \text{rank}(X)\norma{X}$ (see Lemma 11 in Ref.~\cite{theoryQcert} for details).
\end{proof}
Once we have the Choi state $J(\mc{E})$, we project it using the steps outlined in Subsec.~{\ref{alg2det}} and show that the projected Choi state $J_p$ is close to the Choi state $J(\mc{E})$ as follows.
\begin{lemma}[CPTP-projecting a Choi state]\label{regularizeChoi}
\normalfont
Let $J(\mc{E})$ be the Choi state corresponding to some $m$-mode (qubit) channel $\mc{E}$, and let $J(\hat{\mc{E}})$ be the learned version of the Choi state such that $\norma{J(\hat{\mc{E}})-J(\mc{E})}_2^2 \leq \epsilon_1^2$.
The projected Choi state $J_p$ satisfies the bound
\begin{align}
\norma{J( \mc{E})-J_p}_1 \leq \epsilon_l, \label{lemmaregchoi}
\end{align}
where $\epsilon_l = C_0\epsilon_1$ (with $C_0$ being a constant).
\end{lemma}

\begin{proof}
The projection, with respect to the Frobenius norm (defined as $\norma{A}_2=\sqrt{\tr[A^\dagger A]}$), to a trace-preserving map is defined as
\begin{align}
\text{Proj}_{\text{TP}}[X]=\text{argmin}_{X^\prime}\norma{X-X^\prime}_2\\
\text{s.t.~}\tr_A[X^\prime]=\frac{\mathbb{1}}{d_0}.
\end{align}
The unique solution satisfies the inequality
\begin{align}
\norma{\text{Proj}_{\text{TP}}[X] - Y}_2^2 \leq \norma{X-Y}_2^2, \label{firstprojprop}
\end{align}
where $Y$ corresponds to a trace preserving map, and $X$ is an arbitrary matrix.
We note that this unique solution has an exact analytical form given in Proposition 11 of Ref.~\cite{PLS}, which means we can find the projection by computing this expression with a classical computer.
The projection, with respect to the Frobenius norm, to a completely positive map is defined as 
\begin{align}
\text{Proj}_{\text{CP}}[X]=\text{argmin}_{X^\prime}\norma{X-X^\prime}_2^2,\\
\text{s.t.~}X^\prime \geq 0. 
\end{align}
The unique solution satisfies the following inequality: 
\begin{align}
\norma{\text{Proj}_{\text{CP}}[X] - Y}_2^2 \leq \norma{X-Y}_2^2,\label{secprojprop}
\end{align}
where $Y \geq 0$ and $X$ is arbitary. This unique solution has an exact analytical expression which can be found in Proposition 12 of Ref.~\cite{PLS}.

Let us define $J_1:=\text{Proj}_{\text{CP}}[J(\hat{\mc{E}})]$ and $J_2:=\text{Proj}_{\text{TP}}[J_1]$.
Let $\lambda_i$ be the eigenvalues of $J_2$, and let $\lambda_{\text{min}}$ be the minimum eigenvalue. First note that 
\begin{align}
\norma{J_2 - J(\mc{E})}_2^2 &= \norma{\text{Proj}_{\text{TP}}[J_1]-J(\mc{E})}_2^2\\
&\leq \norma{J_1-J(\mc{E})}_2^2 \label{firstmanip}\\
&=\norma{ \text{Proj}_{\text{CP}}[J(\hat{\mc{E}})] - J(\mc{E})}_2^2\\
&\leq \norma{J(\hat{\mc{E}}) - J(\mc{E})}_2^2 \label{secmanip2}\\
&\leq \epsilon_1^2, \label{j2jbound}
\end{align}
where we use Eq.~(\ref{firstprojprop}) in Eq.~(\ref{firstmanip}) and use Eq.~(\ref{secprojprop}) in Eq.~(\ref{secmanip2}).
In the case $\lambda_{\text{min}}\geq 0$, we choose $J_p=J_2$ as the projected Choi state that satisfies $\norma{J_2-J(\mc{E})}_1 \leq d_0^2 \epsilon_1$. In the case $\lambda_{\text{min}}<0$, we set $J_p=J_3$, where $J_3$ is defined as follows:
\begin{align}
J_3=(1-p)J_2 + \frac{p}{d_0^2}\mathbb{1}\otimes \mathbb{1} \label{J3def},\\
(1-p)\lambda_{\text{min}}+\frac{p}{d_0^2}=0 \label{lamcon}.
\end{align}
From the condition in Eq.~(\ref{lamcon}), we have $J_3 \geq 0$ (resulting in complete positivity of the corresponding channel). Taking the partial trace of $J_3$ over the first subsystem, which we denote as system $A$, gives us $\tr_A[J_3]=\mathbb{1}/d_0$ (we use the fact that $\tr_A[J_2]=\mathbb{1}/d_0$). Before we bound $\norma{J_3-J(\mc{E})}_1$, we gather here a few facts. 
Since $J(\mc{E})$ corresponds to a CPTP state, it has eigenvalues between 0 and 1. Using Weyl's perturbation theorem in Theorem \ref{weylth} and the fact that $\norma{J_2 - J(\mc{E})} \leq \norma{J_2 - J(\mc{E})}_2 \leq \epsilon_1$ gives us $\norma{J_2}\leq 1+\epsilon_1$ and $\lambda_{\text{min}} \geq -\epsilon_1$. Using these facts with Eq.~(\ref{lamcon}) gives us
\begin{align}
p=\frac{-\lambda_{\text{min}}}{1/d_0^2-\lambda_{\text{min}}} \leq d_0^2 \epsilon_1,\\
\norma{J_2} \leq 1+\epsilon_1.
\end{align}
Using the above results and the definition of $J_3$ from Eq.~(\ref{J3def}) gives us
\begin{align}
\norma{J_3 - J_2} &\leq p \norma{J_2} + \frac{p}{d_0^2}\\
&\leq p(1+\epsilon_1)+ \frac{p}{d_0^2}\\
&\leq 3d_0^2 \epsilon_1, \label{j3j2}
\end{align}
where we assume $\epsilon_1 <1$. We can then finally bound $\norma{J_3-J(\mc{E})}_1$ as follows: 
\begin{align}
\norma{J_3 -J(\mc{E})}_1 &\leq \norma{J_3-J_2}_1 + \norma{J_2-J(\mc{E})}_1 \notag \\
&\leq d_0^2 \norma{J_3 - J_2}+d_0^2 \norma{J_2-J}\notag \\
&\leq d_0^2 \norma{J_3 - J_2}+d_0^2 \norma{J_2-J}_2\notag\\
& \leq (3d_0^4 + d_0^2)\epsilon_1 \notag\\
&=C_0 \epsilon_1 \notag\\
&=\epsilon_l,
\end{align}
where we use the triangle inequality, Eqs.~(\ref{j2jbound}) and (\ref{j3j2}), and the property $\norma{A}_1 \leq \text{rank}(A)\norma{A}$.
\end{proof}
Since the projected Choi state $J_p$ is close to the Choi state $J(\mc{E})$, we can prove the following distance bound between the channel $\mc{E}^{\mc{Q}}_m$ (corresponding to the Choi state $J(\mc{E})$) and the channel $\mc{E}^{ \mc{Q}}_{m,\text{proj}} $ (corresponding to the Choi state $J_p$).
\begin{corollary}{(\normalfont Distance between the learned channel and the projected channel)} \label{distanceLearned}
\normalfont
We can obtain the following bound on the diamond distance between the channel $\mc{E}^{ {\mc{Q}}}_m$ and the channel $\mc{E}^{ {\mc{Q}}}_{m,\text{proj}} $: 
\begin{align}
\mc{D}_{\diamond}( \mc{E}^{ {\mc{Q}}}_m, \mc{E}^{ {\mc{Q}}}_{m,\text{proj}} ) \leq C_3\epsilon_2, \label{edem}
\end{align}
where $C_3 =   d_0^{11}(3d_0^2+1)/2$, $\epsilon_2 = \text{max}_{\alpha, \beta}\absa{\hat q_{\alpha \beta} - q_{\alpha \beta}}$, $d_0=2^m$,  and $m=\kappa t/2$. Here $\mc{E}^{ {\mc{Q}}}_{m,\text{proj}} $ is obtained by projecting the learned Choi state of the channel $\mc{E}^{ {\mc{Q}}}_m$ from Lemma \ref{regularizeChoi}.
\end{corollary}
\begin{proof}
This result follows from the following computation: 
\begin{align}
\mc{D}_{\diamond}( \mc{E}^{ {\mc{Q}}}_m, \mc{E}^{ {\mc{Q}}}_{m,\text{proj}}) &\leq \frac{d_0}{2}\norma{J(\mc{E}) - J_p}_1\label{choidiam}\\ 
&\leq \frac{d_0}{2}C_0 \epsilon_1. \label{feqex}
\end{align}
In Eq.~(\ref{choidiam}), we use the following bound from Lemma 26 in Ref.~\cite{Unification}: 
\begin{align}
\mc{D}_{\diamond}(\mc E_1, \mc E_2) \leq \frac{d_0}{2}\norma{J(\mc{E}_1)-J(\mc{E}_2)}_1.
\label{diamondChoi}
\end{align}
In Eq.~(\ref{feqex}), we use Eq.~(\ref{lemmaregchoi}) from Lemma \ref{regularizeChoi}, where $\norma{J(\mc{E})-J_p}_2 \leq \epsilon_1$ and $\norma{J(\mc{E})-J_p}_1 \leq C_0\epsilon_1$. Now note that $\norma{J(\mc{E})-J_p}_2 \leq d_0^2 \norma{J(\mc{E})-J_p} \leq  d_0^8 \epsilon_2$, where we use Eq.~(\ref{jchoinorm}) and $\epsilon_2 = \text{max}_{\alpha, \beta}\absa{\hat q_{\alpha \beta} - q_{\alpha \beta}}$. Therefore, we can choose $\epsilon_1 = d_0^8 \epsilon_2$ to get
\begin{align}
\mc{D}_{\diamond}( \mc{E}^{ {\mc{Q}}}_m, \mc{E}^{ {\mc{Q}}}_{m,\text{proj}}) &\leq \frac{d_0^9}{2}C_0 \epsilon_2
\end{align}
and the result in Eq.~(\ref{edem}).
\end{proof}

We now proceed to provide the learning guarantee for our algorithm in the qubit implementation. The analogous result for the fermionic implementation is provided in Appendix \ref{fermionicimp}.
\begin{lemma}[Learning algorithm guarantee for the qubit implementation] \label{guaranteeA}
\normalfont
Let $U_t$ be the unknown unitary defined in Eq.~(\ref{promiseUt}) with a qubit implementation. There is a learning algorithm that learns the unknown unitary as the $m$-qubit channel $\mc{E}^{\bar{\mc{W}_t}}_{m,\text{proj}}$ satisfying the distance bound  
\begin{align}
\mc{D}_{\diamond}(\bar{\mc{W}}_t,\mc{E}_{m, \text{proj}}^{\bar{\mc{W}_t}} \otimes \mc I_B ) \leq T_2(n) \epsilon
\end{align}
with probability $\geq 1-\delta$, using $O(\text{poly}(n, \epsilon^{-1}, \log\small \delta^{-1}))$ accesses to $U_t$ and $O(\text{poly}(n, \epsilon^{-1}, \log\small \delta^{-1}))$ classical processing time. Here $T_2(n)=\text{poly}(n)$.
\end{lemma}
\begin{proof}
Running Algorithm \ref{alg:algo1} with input parameters $(\epsilon^2, \delta/2)$ and some postprocessing gives the reduced channel $\mc{E}_m^{\bar{\mc{W}_t}}$ with the bound $\mc{D}_{\diamond}( \bar{\mc{W}}_t, \mc{E}_m^{\bar{\mc{W}_t}} \otimes \mc I_B ) \leq n(2n+3)T_1(n) \epsilon$, where $\bar{\mc{W}}_t$ is the unitary channel corresponding to the unitary $\bar W_t$, with probability $\geq 1-\delta/2$. This follows from the following computation:
\begin{align}
\mc{D}_{\diamond}( \bar{\mc{W}}_t, \mc{E}_m^{\bar{\mc{W}_t}} \otimes \mc I_B ) 
& \leq n(2n+3)\epsilon_0\label{grnt2}\\
&\leq n(2n+3)T_1(n)\epsilon,\label{grnt3}
\end{align}
where line (\ref{grnt2}) follows from Eq.~(\ref{eqcdec}) in Lemma \ref{localglobal} and Eq.~(\ref{paulidec}). To obtain Eq.~(\ref{grnt3}), we use Eq.~(\ref{epsilon0}) from Lemma \ref{constructW}. Running Algorithm \ref{alg:algo2} with input parameters $(\epsilon, \delta/2)$ to learn the reduced quantum channel $\mc{E}^{\bar{\mc{W}_t}}_m$ and projecting using our scheme gives us the following bound between the channels $\mc{E}^{\bar{\mc{W}_t}}_m$ and the projected channel $\mc{E}^{\bar{\mc{W}_t}}_{m,\text{proj}}$ from Corollary \ref{distanceLearned}:
\begin{align}
\mc{D}_{\diamond}(\mc{E}^{\bar{\mc{W}_t}}_m, \mc{E}^{\bar{\mc{W}_t}}_{m,\text{proj}}) \leq C_3 \epsilon,
\end{align}
with probability $\geq 1-\delta/2$. Here $C_3$ is a constant defined in Eq.~(\ref{edem}). We can then use the triangle inequality to obtain the channel distance bound between the channel $\bar{\mc{W}_t}$ and the projected version of the learned channel $\mc{E}^{\bar{\mc{W}_t}}_{m,\text{proj}}$ as follows:
\begin{align}
&\mc{D}_{\diamond}(\bar{\mc{W}}_t,\mc{E}_{m, \text{proj}}^{\bar{\mc{W}_t}} \otimes \mc I_B ) \notag\\ \leq &\mc{D}_{\diamond}(\bar{\mc{W}}_t,\mc{E}_{m}^{\bar{\mc{W}_t}} \otimes \mc I_B )+\mc{D}_{\diamond}(\mc{E}_{m}^{\bar{\mc{W}_t}}, \mc{E}_{m, \text{proj}}^{\bar{\mc{W}_t}} )\\
\leq & T_2(n)\epsilon, \label{almostbound}
\end{align}
where $T_2(n)=n(2n+3)T_1(n)+C_3=\text{poly}(n)$. From the union bound, the algorithm succeeds with probability $\geq 1-\delta$. From the $\epsilon$-dependence of the number of states required for Algorithms 1 and 2, the learning algorithm 
uses $N_c + \bar N_c = O(\text{poly}(n,\epsilon, \log \delta^{-1}))$ accesses of the unknown unitary $U_t$ to achieve error $\epsilon$ in Eq.~(\ref{almostbound}), where $N_c$ and $\bar N_c$ are defined in Algorithms 1 and 2, respectively. Moreover, each step of the learning algorithm requires $\text{poly}(n, \epsilon^{-1}, \log \delta^{-1})$ classical processing time.
\end{proof}

\section{Technical lemmas for the fermionic implementation}\label{fermionicimp}
In this section, we prove key technical lemmas for the fermionic implementation. 

We first prove Lemma \ref{lastlemma}, which shows that,  in the fermionic implementation, we can construct $W_t$ such that the condition $[W_t, \gamma_i]\approx 0$ for $i>M$ implies that $W_t$  is Pauli decoupled from modes $i>m$.

\lastlemma*
\begin{proof}
In the case where we are given that $G_{t^\prime}$ in $U_t$ correspond to SO$(2n)$, it follows from Lemmas \ref{Udecomp} and \ref{constructW} that $W_t=G_a^\dagger U_t G_b^\dagger$ obtained from Algorithm \ref{alg:algo1} has the form $G_1 u_t G_2$, where both $G_1$ and $G_2$ correspond to orthogonal matrices in SO$(2n)$ since both $G_a$ and $G_b$ can be chosen to correspond to SO$(2n)$. This means that $W_t$ is a sum of Majorana strings of even weight and satisfies Eq.~(\ref{approxmaj}) in Lemma \ref{constructW}. We note that $W_t$ can be written as 
\begin{align}
W_t = W_t^{\text L} + W_t^{\text{NL}},
\end{align}
where $W_t^{\text{L}}$ is  supported on the first $M$ Majorana operators, and $W_t^{\text{NL}}$ contains Majorana strings $\tilde \gamma_x$ containing at least one Majorana operator $\gamma_i$ with $i > M$. We first introduce the following notation. Let $f_j$ and $\bar f_j$ be functions defined on operators as follows: 
\begin{align}
f_j(X)&=\frac{1}{2}[X, \gamma_j]\gamma_j,\\
\bar f_j(X)&=\frac{1}{2} \{X, \gamma_j \}\gamma_j,
\end{align}
where $X$ is any operator. For any operator $X$, we can always write $X=X_j+\bar X_j$, where $\{X_j, \gamma_j \}=0$, and $[\bar X_j, \gamma_j]=0$. This is because $X$ can be written as a sum of Majorana strings $\tilde \gamma_x$, and each $\tilde \gamma_x$ either commutes or anticommutes with $\gamma_j$ (since $\tilde{\gamma}_x$ and $\gamma_j$ are in the Pauli group). The function $f_j$ then satisfies $f_j(X)=X_j$ and $\bar f_j(X)=\bar X_j$. Additionally, we have the following result: 
\begin{align}
\lVert f_j \bar f_k(X) \rVert &\leq \lVert f_j(X) \rVert, \label{fres2}
\end{align}
for any $X$ and $j \neq k$. This follows from the triangle inequality and the following computation: 
\begin{align}
f_j \circ \bar f_k(X)&=\frac{1}{2}[\bar f_k(X), \gamma_j]\gamma_j \notag\\
&=\frac{1}{4}[X+\gamma_k X \gamma_k, \gamma_j]\gamma_j \notag\\
&=\frac{1}{4}\Big( [X, \gamma_j]\gamma_j + \gamma_k [X, \gamma_j]\gamma_j \gamma_k \Big) \notag\\
&=\frac{1}{2} (f_j(X)+\gamma_k f_j(X)\gamma_k).
\end{align}
We first write $W_t^{\text{NL}}$ as follows:
\begin{align}
W_t^{\text{NL}} = f_{M+1}(W_t^{\text{NL}}) + \sum_{k=M+2}^{2n}f_k   \bar f_{k-1}  \bar f_{k-2}   \cdots   \bar f_{M+1}(W_t^{\text{NL}}), \label{Xdecomp}
\end{align}
where we use the notation $\bar f_{i_1} \circ f_{i_2} \circ\cdots \circ f_{i_n}(X)=:\bar f_{i_1}f_{i_2}\hdots f_{i_n}(X)$. The decomposition in Eq.~(\ref{Xdecomp}) follows from the following computation:
\begin{align}
&W_t^{\text{NL}} \notag\\
=& f_{M+1}(W_t^{\text{NL}}) + \bar f_{M+1}(W_t^{\text{NL}}) \notag\\
=&f_{M+1}(W_t^{\text{NL}})+f_{M+2} \bar f_{M+1}(W_t^{\text{NL}}) + \bar f_{M+2} \bar f_{M+1}(W_t^{\text{NL}}) \notag\\
=& f_{M+1}(W_t^{\text{NL}})+f_{M+2} \bar f_{M+1}(W_t^{\text{NL}}) +\notag
\\& \cdots+ f_{2n} \bar f_{2n-1}  \cdots    \bar f_{M+1}(W_t^{\text{NL}}),
\end{align}
where we use the fact that all terms in $W_t^{\text{NL}}$ contain at least one $\gamma_i$ with $i > M$.
We now proceed to bound $W_t^{\text{NL}}$ as follows: 
\begin{align}
&\norma {W_t^{\text{NL}}} \notag \\ \leq& \norma{f_{M+1}(W_t^{\text{NL}})} + \notag \sum_{k=M+2}^{2n} \norma{f_k   \bar f_{k-1} \bar f_{k-2}  \cdots   \bar f_{M+1}(W_t^{\text{NL}})}\\
\leq& \norma{f_{M+1}(W_t^{\text{NL}})} + \sum_{k=M+2}^{2n}\norma{f_k(W_t^{\text{NL}})} \notag\\
\leq & (2n-M)\max_{k >M} \norma{f_k(W_t^{\text{NL}})}, \notag
\end{align}
where we use the fact that $\norma{f_{i_1} \bar f_{i_2}\cdots  \cdots \bar f_{i_n} (X)} \leq \norma{f_{i_1}(X)} $, which in turn follows from using Eq.~(\ref{fres2}) repeatedly. Now note that we can write $f_k(W_t^{\text{NL}}) = f_k(W_t)$ since $W_t^{\text{L}}$ commutes with $\gamma_k$ for $k>M$ (since $W_t$ contains Majorana strings $\tilde \gamma_x$ with even weight). Then the result $\norma{[W_t, \gamma_k]} \leq \epsilon_0$ from Eq.~(\ref{approxmaj}) gives us $\norma{f_k(W_t)} \leq \epsilon_0/2$ and the following result:
\begin{align}
\norma{W_t^{\text{NL}}} \leq n\epsilon_0.
\end{align}
The above equation then gives us the result in Eq.~(\ref{lastlemmares}) using the fact that, for $i>m$, $\norma{[W_t, P_i]} =\norma{[W_t^{\text{NL}}, P_i]} \leq 2 \norma{W_t^{\text{NL}}} \norma{P_i} \leq 2n\epsilon_0$.
\end{proof}
We now show how to modify the learning algorithm for the fermionic implementation. For Algorithms \ref{alg:algo1} and \ref{alg:algo2}, we modify the states (and the corresponding observables) used in the shadow tomography protocols to ones that can be prepared on a fermionic computer (i.e., states that can be prepared by a parity-preserving quantum circuit). Throughout this section, we will often describe our fermionic unitaries as parity-preserving qubit unitaries, related to the actual fermionic ones we have in mind through the Jordan-Wigner transformation.

\subsection{Fermionic implementation: states and observables for Algorithm \ref{alg:algo1}}

In Lemma \ref{lemcx}, we perform state tomography on some qubit state $\ket{\psi_j}$ (where the index $j$ corresponds to the bitstring $x$ with weight 1 and $x_j=1$) to estimate physical observables $O_k^{+}$ and construct the matrix $c^{(1)}$. For the fermionic implementation, we modify the states and observables as follows. We use the mapping between fermionic and qubit states where any computational basis state on qubits $\ket{z_1 \hdots z_n}$ is identified with the corresponding fermionic state in the occupation basis. We prepare the fermionic state $\ket{\psi_j^\text{f}}$ defined as
\begin{align}
\ket{\psi_j^{\text{f}}}&=(U_t\otimes I) \ket{\phi_j^\text{f}}, \label{psijfdef}\\
\ket{\phi_j^{\text{f}}}&=\frac{1}{\sqrt{2}}U_j^{\text{f}}(\ket{00}-\ket{11})_{\mc{A}_1 \mc{A}_2 }\ket{\Phi_d} \label{parityqubit},
\end{align}
where $\ket{\Phi_d} \propto \sum_{z\in \{0,1 \}^n} \ket{z,z}$, where $\ket{z}$ are the occupation basis states on $n$ modes. $\ket{\phi^{\text{f}}_j}$ is a fermionic state defined on modes $\mc{A}_1, \mc{A}_2, 1, \hdots, 4n$, and $U_j^{\text{f}}$ is a fermionic unitary defined as $U_j^{\text{f}}=(1-a_{\mc{A}_2}^\dagger a_{\mc{A}_2})+a_{\mc{A}_2}^\dagger a_{\mc{A}_2} \gamma_{\mc{A}_1} \gamma_j$, where $j\in \{1, \hdots, 2n\}$. The unitary $U_t$ acts on modes labeled $1, \hdots, n$. We write $U_t \otimes I$ to emphasize the fact that $U_t$ only acts on modes $1, \hdots, n$. The state $\ket{\phi_j^{\text{f}}}$ can be simplified as follows:
\begin{align}
\ket{\phi_j^{\text{f}}}=\frac{1}{\sqrt{2}}(\ket{00}\ket{\Phi_d}+a_{\mc{A}_2}^\dagger \gamma_j \ket{00}\ket{\Phi_d}).
\end{align}
We note here that the state $(\ket{00}-\ket{11})
\ket{\Phi_d}$ can be prepared by a fermionic circuit efficiently using the facts that the qubit version of this state can be prepared by a circuit composed of parity-preserving two-qubit gates as shown in Fig.~\ref{stateprep}, and that any two-qubit parity-preserving gate can be implemented as a series of fermionic gates \cite{bravyi}.

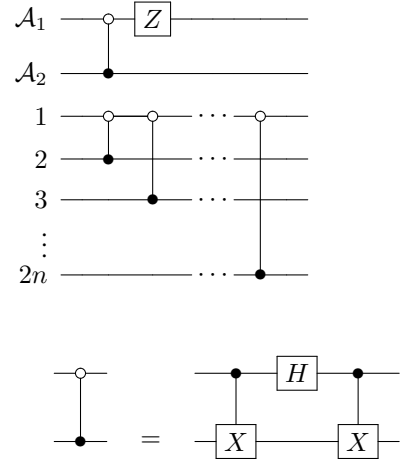
\begin{figure}[t]
\centering
\hspace{0cm}
\Qcircuit @C=0.8em @R=1.3em {
\lstick{\mathcal{A}_1}& \qw  & \ctrlo{1} & \gate{Z} & \qw & \qw &\qw &\qw&\qw &\qw    \\
\lstick{\mathcal{A}_2}& \qw    & \control \qw& \qw &  \qw& \qw & \qw & \qw & \qw &  \qw  \\
\lstick{1}& \qw & \ctrlo{1}  & \ctrlo{2} \qw &  \qw & \cdots &    & \ctrlo{4}&  \qw  & \qw\\
\lstick{2}& \qw & \control \qw & \qw & \qw &  \cdots &    &   \qw & \qw & \qw  \\
\lstick{3}& \qw  & \qw &  \control \qw & \qw  & \cdots & &  \qw & \qw& \qw \\
\lstick{\vdots}&    &   & &    &   &   &  &     &  &   &    &   &    \\
\lstick{2n}& \qw  & \qw & \qw& \qw  &   \cdots &   & \control \qw&  \qw & \qw  \\}\\
\vspace{1cm}
\hspace{0cm}
\Qcircuit @C=0.8em @R=1.3em {
& \ctrlo{1} & \qw & & & \ctrl{1} &\gate{H}  & \ctrl{1} & \qw  \\
& \control \qw &\qw & \push{\rule{.2em}{0em}=\rule{.2em}{0em}} &  &\gate{X}  &\qw & \gate{X} & \qw
}
\caption{The qubit unitary, composed of parity-preserving two-qubit unitaries, needed to prepare the qubit state $(\ket{00}-\ket{11})\ket{\Phi_d}$ in Eq.~(\ref{parityqubit}). This implies that this state, now thought of as a fermionic state in the occupation basis, can be prepared on a fermionic quantum computer \cite{bravyi}.}
\label{stateprep}
\end{figure}

The observables in Eqs.~(\ref{Oplus}) and (\ref{Ominus}) are redefined to 
\begin{align}
O_k^+&=(a_{\mc{A}_2}^\dagger - a_{\mc{A}_2})\gamma_k, \label{oplusferm}\\
O_k^-&=i(a_{\mc{A}_2}^\dagger + a_{\mc{A}_2})\gamma_k,
\end{align}
giving us $O_k=a_{\mc{A}_2}^\dagger \gamma_k=(O_k^+ - iO_k^-)/2$ in place of Eq.~(\ref{okqbit}) with the desired expectation value in Eq.~(\ref{expecs1}). We measure only $O_k^+$ since $O_k\propto (O_k^+ - iO_k^-)$ and $c_{jk}^{(1)}\propto \tr[\rho_j O_k]$ (since $c^{(1)}_{jk}$ is a real matrix).
The shadow tomography step can be implemented on the fermionic platform by first applying a unitary from the fermionic Gaussian ensemble, performing a basis measurement in the occupation basis, and then performing classical post-processing on a classical computer (e.g.~by representing Majorana operators via the Jordan-Wigner encoding). Since the state is defined on a different number of modes, the shadow tomography step in Algorithm \ref{alg:algo1} needs
\begin{align}
N_c^{\text{f}}=\Big(1+\frac{\epsilon}{6n} \Big)\log(8n^2/\delta)\frac{4n^2(4n+3)}{\epsilon^2} \label{expncf}
\end{align}
copies of state $\ket{\psi_j^{\text{f}}}$ (for more details, see discussion surrounding Eq.~(\ref{zharores}) in Lemma \ref{shadowestf}). This step gives us the matrix $c^{(1)}$,  which can be processed, as described in Lemma \ref{constructW}, to give us the unitaries $G_a$, $G_b$, and the unitary $W_t=G_a^\dagger U_t G_b^\dagger$ that satisfies the Majorana decoupling property in Eq.~(\ref{approxmaj}).

\subsection{Fermionic implementation: states and observables for Algorithm \ref{alg:algo2}}

We now modify the states and observables in Algorithm \ref{alg:algo2} used to construct the Choi state $J(\hat{\mc{E}})$ for the reduced quantum channel corresponding to the unitary $W_t$. We define states on the fermionic modes $\mc{A}_1, \mc{A}_2, 1, \hdots, 2n$, where $\mc{A}_1$, $\mc{A}_2$ are ancilla modes. Instead of using the state defined in Eq.~(\ref{psialpha}), we use the state
\begin{align}
\lvert\bar \psi_\alpha^\text{f} \rangle &= (W_t\otimes I)\ket{\bar{\phi}^{\text{f}}_\alpha}, \label{psialphaf}\\
\ket{\bar{\phi}_\alpha^{\text{f}}}&=\frac{1}{\sqrt{2}}\bar U_{\alpha}^{\text{f}p}\big( \ket{00}_{\mc{A}_1 \mc{A}_2} + \ket{11}_{\mc{A}_1 \mc{A}_2})\ket{\Phi_d},
\end{align}
where $\alpha\in\{I,X,Y,Z\}^{\otimes m}$ and $W_t$ acts on modes $1, \hdots, n$. Here $p=0$ when the Pauli observable $\bar{P}_{\alpha}$ is parity-preserving and $p=1$ when it is not. The unitary $\bar U_{\alpha}^{\text{f}p}$ is defined as follows:
\begin{align}
\bar U_{\alpha}^{\text{f}p} = (\ketbra{0}{0}_{\mc{A}_1}\otimes I + \ketbra{1}{1}_{\mc{A}_1}\otimes (X_{\mc{A}_2})^p \bar P_\alpha),
\end{align}
where $\bar P_{\alpha}$ acts on modes $1, \hdots, m$. We note that, due to the Jordan-Wigner transformation,  $X_{\mc{A}_2}$ acts on both modes $\mc{A}_2$ and $\mc{A}_1$.
We also note here that the qubit version of state $(\ket{00}_{\mc{A}_1 \mc{A}_2}+\ket{11}_{\mc{A}_1 \mc{A}_2})\ket{\Phi_d}$ can be prepared using a circuit composed of parity-preserving two-qubit gates (similar to the circuit in Fig.~\ref{stateprep}). Moreover, the unitary $\bar U_{\alpha}^{\text{f}p}$ is a parity-preserving gate, ensuring that the state $\lvert\bar \psi_\alpha^\text{f} \rangle$ can be implemented on a fermionic quantum computer. We can then employ shadow tomography using the local Clifford unitary ensemble to estimate the expectation value of the operator $\bar O_\beta^+$ defined as
\begin{align}
\bar O_\beta^+ &= X_{\mc{A}_1} X_{\mc{A}_2}\otimes \bar P_\beta \label{obsf1}
\end{align}
for $p=0$. 
 For the case where $p=1$, we use the observable $\bar O_\beta^+$ defined as
\begin{align}
\bar O_\beta^+ = X_{\mc{A}_1}Z_{\mc{A}_2}\otimes \bar P_\beta. \label{obsf2}
\end{align}
The number of copies required for shadow tomography using the local Clifford ensemble (for the cases $p=0$ and $p=1$) is given by
\begin{align}
\bar N_c^{\text{f}} = C_1^\prime \log(C_2^\prime/\delta)/\epsilon^2, \label{ncfbar}
\end{align}
where $C_1^\prime = 68(3^{m+2})$  and $C_2 = 2(4^{2m})$.
The choice of observables in Eqs.~(\ref{obsf1}) and (\ref{obsf2}) ensures that $\tr\big[\bar O_\beta^+|\bar \psi_\alpha^{\text{f}}\rangle \langle \bar \psi_\alpha^{\text{f}}|\big] = f_{\alpha \beta}$, where $f_{\alpha \beta}$ is defined in Eq.~(\ref{paulilearneq}) for $S=W_t$, and $\alpha, \beta \in \{I, X, Y, Z \}^{\otimes m}$.

We now show how to implement any unitary $U_c$ from the local Clifford ensemble for the shadow tomography step. First, consider the case where $U_c$ is parity-preserving. In that case, we can implement the gate built from fermionic gates \cite{bravyi}. For the case where $U_c$ doesn't preserve parity, as shown in Ref.~\cite{bravyi}, we can define the following gate:
\begin{align}
\tilde{U}_c = V_p^\dagger(I\otimes U_c)V_p, \label{defuc}
\end{align}
using an ancilla mode labeled 0, where $V_p \ket{z_0, z_1, \cdots, z_N}=\ket{z_p, z_1, \hdots, z_N}$ with $z_p = z_0+z_1+\hdots+z_N$, making $\tilde U_c$ parity-preserving. We can then decompose $\tilde{U}_c$ into a product of two-qubit parity-preserving gates, each of which can be implemented on a fermionic quantum computer using a series of fermionic gates \cite{bravyi}.

By implementing Algorithm \ref{alg:algo2}, we can construct the projected Choi state $J_p$ of the reduced channel corresponding to ${\mc{W}}_t$. In the final step,  where the Stinespring dilation $V_S$ (using $2m$ ancilla fermionic modes) is constructed from the projected Choi state $J_p$, a parity-preserving unitary $\tilde V_S$ can be constructed using the same trick that defines $\tilde U_c$ in Eq.~(\ref{defuc}). $\tilde V_S$ can then be used to implement the reduced quantum channel as shown in Fig.~\ref{cptp}.

\begin{lemma}[Learning algorithm guarantees for the fermionic implementation] \label{guaranteeAf}
\normalfont
Let $U_t$ be the unknown unitary defined in Eq.~(\ref{promiseUt}) in a fermionic implementation. There is a learning algorithm that learns the unknown unitary as the $m$-mode channel $\mc{E}^{{\mc{W}_t}}_{m,\text{proj}}$ satisfying the  distance bound 
\begin{align}
\mc{D}_{\diamond}( {\mc{W}}_t,\mc{E}_{m, \text{proj}}^{{\mc{W}_t}} \otimes \mc I_B ) \leq T_2^{\text{f}}(n) \epsilon
\end{align}
with probability $\geq 1-\delta$, using $O(\text{poly}(n, \epsilon^{-1}, \log \delta^{-1}))$ accesses to $U_t$ and $O(\text{poly}(n, \epsilon^{-1}, \log \delta^{-1}))$ classical processing time. Here $T^{\text{f}}_2(n)=\text{poly}(n)$.
\end{lemma}
\begin{proof}
Running Algorithm \ref{alg:algo1} with input parameters $(\epsilon^2, \delta/2)$ and some postprocessing gives the reduced channel $\mc{E}^{{\mc{W}_t}}$ with the bound $\mc{D}_{\diamond}({\mc{W}}_t, \mc{E}^{{\mc{W}_t}} \otimes \mc I_B ) \leq 3n^2T_1(n) \epsilon$, where  ${\mc{W}}_t$ is the unitary channel corresponding to the unitary $ W_t$, with probability $\geq 1-\delta/2$. This follows from the following computation:
\begin{align}
\mc{D}_{\diamond}({\mc{W}}_t, \mc{E}_m^{{\mc{W}_t}} \otimes \mc I_B ) 
& \leq3n^2 \epsilon_0\label{grntf1}\\
&\leq 3n^2 T_1(n)\epsilon,\label{grntf2}
\end{align}
where line (\ref{grntf1}) follows from Eq.~(\ref{eqcdec}) in Lemma \ref{localglobal} and Eq.~(\ref{lastlemmares}) in Lemma \ref{lastlemma}. Line (\ref{grntf2}) follows from Eq.~(\ref{epsilon0}) in Lemma \ref{constructW}. Running Algorithm \ref{alg:algo2} with input parameters $(\epsilon, \delta/2)$ to learn the reduced quantum channel and then projecting it to a CPTP map gives us the following bound on the distance between the channel  $\mc{E}^{{\mc{W}_t}}_m$ and the projected channel $\mc{E}^{{\mc{W}_t}}_{m,\text{proj}}$ from Corollary \ref{distanceLearned}:
\begin{align}
\mc{D}_{\diamond}(\mc{E}^{{\mc{W}_t}}_m, \mc{E}^{{\mc{W}_t}}_{m,\text{proj}}) \leq C_3 \epsilon,
\end{align}
with probability $\geq 1-\delta/2$. Here $C_3$ is a constant defined in Eq.~(\ref{edem}). We can then use the triangle inequality to obtain the channel distance bound between the channel ${\mc{W}_t}$ and the projected version of the learned channel $\mc{E}^{{\mc{W}_t}}_{m,\text{proj}}$ as follows: 
\begin{align}
&\mc{D}_{\diamond}({\mc{W}}_t,\mc{E}_{m, \text{proj}}^{{\mc{W}_t}} \otimes \mc I_B ) \notag\\ \leq &\mc{D}_{\diamond}({\mc{W}}_t,\mc{E}_{m}^{{\mc{W}_t}} \otimes \mc I_B )+\mc{D}_{\diamond}(\mc{E}_{m}^{{\mc{W}_t}}, \mc{E}_{m, \text{proj}}^{{\mc{W}_t}} )\\
\leq & T_2^{\text{f}}(n)\epsilon, \label{almostboundf}
\end{align}
where $T_2^{\text{f}}(n)=3n^2 T_1(n)+C_3=\text{poly}(n)$. From the union bound, the algorithm succeeds with probability $\geq 1-\delta$. From the $\epsilon$-dependence of the number of states required for Algorithms 1 and 2, the learning algorithm 
uses $N_c^{\text{f}} + \bar N_c^{\text{f}} = O(\text{poly}(n, \epsilon^{-1}, \log \delta^{-1}))$ accesses to the unknown unitary $U_t$ to achieve error $\epsilon$ in Eq.~(\ref{almostboundf}), where $N_c^{\text{f}}$ and $\bar N_c^{\text{f}}$ are defined in Algorithms \ref{alg:algo1} and \ref{alg:algo2}, respectively. Moreover, each step of the learning algorithm requires $\text{poly}(n, \epsilon^{-1}, \log \delta^{-1})$ classical processing time.
\end{proof}

\section{Fermionic unitaries and the matchgate hierarchy}
\label{matchgateApp}
In this section, we state and prove Lemma \ref{outsideMH}, which shows that, generally, fermionic unitaries with just two non-Gaussian gates lie outside the matchgate hierarchy.

\outsideMH*
\begin{proof}
We assume that $U_t$ is in some finite level, say $k$, of the matchgate hierarchy. Then the unitary $F_1 = U_t \gamma_\mu U_t^\dagger$ must be in $\mc{M}_{k-1}$ from the definition of the matchgate hierarchy, where $\mu \in [2n]$. We define the unitaries
\begin{align}
F_{j}:=F_{j-1}\gamma_{\mu} F_{j-1}^\dagger, \>\>\> j \geq 2.
\end{align}
Extending the same argument as above shows that $F_{k-2}$ must belong to $\mc{M}_{k-2}$. For $U_t = K G(\theta)K$ and $\mu=2$, the following results hold: 
\begin{align}
F_1 &= -\gamma_2 (\cos (2\theta)+i\sin(2\theta)\gamma_2 \gamma_3 \gamma_4 \gamma_5),\\
F_j &= \gamma_2 (\cos (2^j\theta) + i\sin (2^j \theta)\gamma_2 \gamma_3 \gamma_4 \gamma_5),\>\>\> j \geq 2.
\end{align}
Choosing $\theta=\pi/p$, where $p$ is an odd integer, shows that $F_{j}$ always has a Majorana string with weight $>1$. This means that $F_{k-2}$ is not Gaussian because $F_{k-2}\gamma_2 F_{k-2}^\dagger = F_{k-1}$ has a Majorana weight $>1$, proving the claim that $U_t$ does not belong to any finite level of the matchgate hierarchy.
\end{proof}

\end{document}